\documentclass[11pt,reqno]{article}
\usepackage{epsf, amsmath, amssymb, graphicx, epsfig, hyperref, mathtools,amsthm,bbm,xcolor,pifont}  % fullpage
\usepackage[utf8]{inputenc}
\usepackage[margin=1in]{geometry}
\usepackage{verbatim}
\usepackage{caption,subcaption}
\usepackage{mathrsfs}
\usepackage{enumitem}
\usepackage{booktabs, caption, makecell}

\usepackage{threeparttable}
\usepackage{algorithm,algorithmic}
\floatname{algorithm}{Calibration Procedure:}
\usepackage{float}
\restylefloat{table}

\newtheorem{theorem}{Theorem}[section]
\newtheorem{definition}{Definition}[section]
\newtheorem{proposition}{Proposition}[section]
\newtheorem{lemma}{Lemma}[section]

\newtheorem{assump}{Assumption}

\newtheorem{remark}{Remark}[section]
\numberwithin{equation}{section}

\renewcommand{\P}{\mathbb{P}}

\newcommand{\R}{\mathbb{R}}
\newcommand{\E}{\mathbb{E}}
\newcommand{\cE}{\mathcal{E}}

\newcommand{\F}{\mathcal{F}}

\newcommand{\cH}{\mathcal{H}}

\newcommand{\fp}{\mathfrak{p}}

\newcommand{\cG}{\mathcal{G}}

\newcommand{\cP}{\mathcal{P}}

\newcommand{\eps}{\varepsilon}

\newcommand{\ind}{\mathbbm{1}}

\newcommand{\nada}[1]{}
\definecolor{gb}{rgb}{0, 0.2, 0.8}

\title{Mortality and Healthcare: a Stochastic Control Analysis under Epstein-Zin Preferences}
\author{Joshua Aurand\thanks{
		University of Colorado, Department of Applied Mathematics, Boulder, CO 80309-0526, USA, email: \texttt{joshua.aurand@colorado.edu}.} \and Yu-Jui Huang\thanks{
		University of Colorado, Department of Applied Mathematics, Boulder, CO 80309-0526, USA, email: \texttt{yujui.huang@colorado.edu}. Partially supported by National Science Foundation (DMS-1715439) and the University of Colorado (11003573).}}

\begin{document}
	\maketitle
	
\begin{abstract}
This paper studies optimal consumption, investment, and healthcare spending under Epstein-Zin preferences. Given consumption and healthcare spending plans, Epstein-Zin utilities are defined over an agent's random lifetime, partially controllable by the agent as healthcare reduces mortality growth. To the best of our knowledge, this is the first time Epstein-Zin utilities are formulated on a controllable random horizon, via an infinite-horizon backward stochastic differential equation with superlinear growth. A new comparison result is established for the uniqueness of associated utility value processes. In a Black-Scholes market, the stochastic control problem is solved through the related Hamilton-Jacobi-Bellman (HJB) equation. The verification argument features a delicate containment of the growth of the controlled morality process, which is unique to our framework, relying on a combination of probabilistic arguments and analysis of the HJB equation. In contrast to prior work under time-separable utilities, Epstein-Zin preferences facilitate calibration. The model-generated mortality closely approximates actual mortality data in the US and UK; moreover, the efficacy of healthcare can be calibrated  and compared between the two countries. %in broad agreement with empirical studies on healthcare across countries. 
\end{abstract}
	
\textbf{MSC (2010):} 
%49K21, % Optimality conditions:	Problems involving relations other than differential equations
%60J05,  %	Discrete-time Markov processes on general state spaces
%60J27,  %Continuous-time Markov processes on discrete state spaces
%91A13, % Games with infinitely many players
91G10,   	%Portfolio theory
93E20. % Optimal stochastic control
\smallskip

\smallskip


\section{Introduction}\label{sec:intro}
Mortality, the probability that someone alive today dies next year, exhibits an approximate exponential growth with age, as observed by Gompertz \cite{Gompertz25} in 1825. %the early 19th century. 
Despite the steady decline of mortality at all age groups {\it across} different generations, the exponential growth of mortality {\it within} each generation has remained remarkably stable, which is called the Gompertz law. Figure~\ref{fig:cal_US} displays this clearly: in the US, mortality of the cohort born in 1900 and that of the cohort born in 1940 grew exponentially at a similar rate; the latter is essentially shifted down from the former.%\footnote{Similar pattern is observed in various countries, as displayed in Figure~\ref{fig:calibration}.}  

\begin{figure}[h]
\centering
\includegraphics[width=.65\linewidth]{{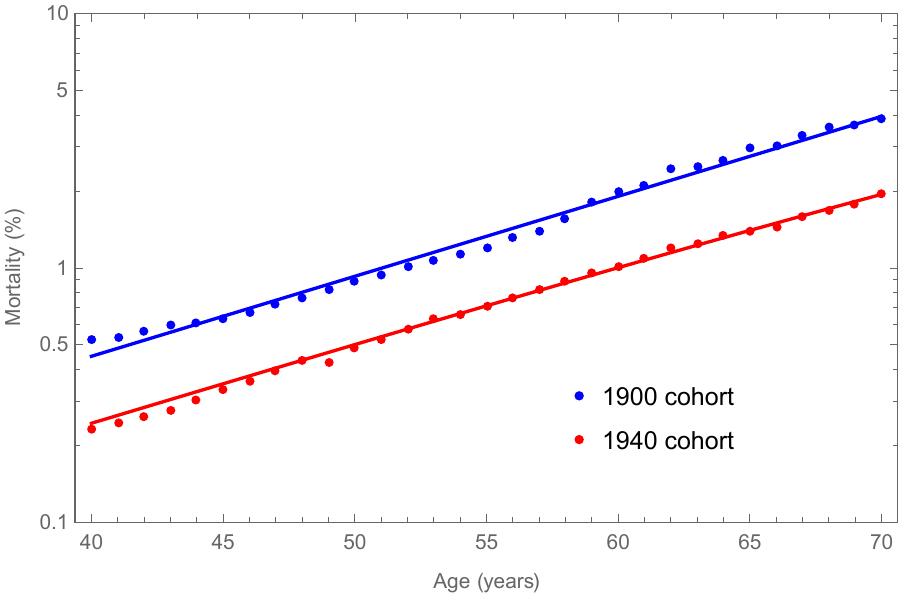}}
%\caption{Mortality rates for the US.}
\caption{\small Mortality rates (log scale) at adults' ages for the cohorts born in 1900 and 1940 in the US. The dots are actual data (Berkeley Human Mortality Database) and the lines are model-implied mortality curves.}
\label{fig:cal_US}
%\label{fig:calibration}
\end{figure}

At the intuitive level, this ``shift down'' %the steady decline 
of mortality across generations can be ascribed to continuous improvement of healthcare and accumulation of wealth. 
Understanding precisely how this ``shift down'' materializes 
%Understanding the precise relations among healthcare, wealth, and mortality 
demands careful modeling in which wealth evolution, healthcare choices, and the resulting mortality are all {\it endogenous}. Standard models of consumption and investment do not seem to serve the purpose: the majority, e.g. \cite{Yaari1965}, \cite{Richard1975}, \cite{rosen1988value}, and \cite{ShepardZeckhauser1984}, consider no more than exogenous mortality, leaving no room for healthcare.\footnote{As an exception, the literature on health capital, initiated by \cite{Grossman1972}, considers endogenous healthcare. Despite its development towards more realistic models, e.g. \cite{EhrlichChuma1990}, \cite{Ehrlich2000}, \cite{Yogo2016}, \cite{HugonnierPelgrinSt-Amour2012}, \cite{hall2007value}, the Gompertz law remains largely absent. 
} 

Recently, Guasoni and Huang \cite{Huang19} directly modeled the effect of healthcare on mortality: healthcare reduces Gompertz' natural growth rate of mortality, through an {\it efficacy function} that characterizes the effect of healthcare spending in a society. Healthcare, as a result, indirectly increases utility from consumption accumulated over a longer lifetime. 
Under the constant relative risk aversion (CRRA) utility function $U(x) = \frac{x^{1-\gamma}}{1-\gamma}$, $0<\gamma<1$, an optimal strategy of consumption, investment, and healthcare spending is derived in \cite{Huang19}, where the constraint $0<\gamma<1$ is justified by interpreting $1/\gamma$ as an agent's {\it elasticity of intertemporal substitution} (EIS). Specifically, to  model mortality endogenously, we need to be cautious of potential preference for death over life. To avoid this, \cite{Huang19} assumes that an agent can leave a fraction $\zeta\in (0,1]$, not necessarily all, of his wealth at death to beneficiaries,  reflecting the effect of inheritance and estate taxes. It is shown in \cite{Huang19} that the optimization problem is ill-posed for $\gamma>1$. Indeed, with $\gamma>1$, or EIS less than one, the income effect of future loss of wealth at death is so substantial that the agent reduces current consumption to zero, leading to the ill-posedness; see below \cite[Proposition 3.2]{Huang19} for details.

Despite the progress in \cite{Huang19}, % in modeling healthcare and mortality endogenously, 
the artificial relation that EIS is the reciprocal of {\it relative risk aversion}, forced by CRRA utility functions, significantly restricts its applications. % of results in \cite{Huang19}. 
Although a preliminary calibration was carried out in \cite[Section 5]{Huang19}, it was not based on the full-fledged model in \cite{Huang19}, but a simplified version without any risky asset. Indeed, once a risky asset is considered, it is unclear whether $\gamma$ should be calibrated to relative risk aversion or EIS. More crucially, empirical studies largely reject relative risk aversion and  EIS being reciprocals to each other: it is widely accepted that EIS is larger than one (see e.g. \cite{Bansal04}, \cite{Bansal07}, \cite{Bhamra10}, and \cite{Benzoni11}), while numerous estimates of relative risk aversion are also larger than one (see e.g. \cite{Vissing03}, \cite{Bansal04}, and \cite{Hansen07}).  

In this paper, we investigate optimal consumption, investment, and healthcare spending under preferences of Epstein-Zin type, which disentangle relative risk aversion (denoted by $0<\gamma\neq 1$) and EIS (denoted by $\psi>0$). In particular, we impose throughout the paper  
\begin{assump}\label{specification}
%\begin{equation}
$\psi >1$ and $\gamma>{1}/{\psi}$. 
%\end{equation}
\end{assump}
\noindent This specification implies a preference for early resolution of uncertainty (as explained in \cite{Skiadas98}), and conforms to empirical estimations mentioned above. 

%In the presence of healthcare, 
Our Epstein-Zin utility process has several distinctive features. First, it is defined on a random horizon $\tau$, the death time of an agent. Prior studies on Epstein-Zin utilities focus on a fixed-time horizon; see e.g. \cite{DuffieLions}, \cite{Schroder96}, \cite{Kraft13},  \cite{Seifried16}, \cite{Kraft17}, and \cite{Xing}. To the best our knowledge, random-horizon Epstein-Zin utilities are developed for the first time in Aurand and Huang \cite{AH19}, where the horizon is assumed to be a stopping time adapted to the market filtration. Our studies complement \cite{AH19}, by allowing for a stopping time (i.e. the death time) that need not depend on the financial market. Second, the random horizon $\tau$ is {\it controllable}: one slows the growth of mortality via healthcare spending, which in turn changes the distribution of $\tau$. Note that a controllable random horizon was considered in a few prior studies, e.g. \cite{Karatzas00} and \cite{DL10}, but all under time-separable utilities. Third, to formulate our Epstein-Zin utilities, we need not only a given consumption stream $c$ (as in the literature), but also a specified healthcare spending process $h$. Given the pair $(c,h)$, the Epstein-Zin utility is defined as the right-continuous process $\widetilde V^{c,h}$ that satisfies a random-horizon dynamics (i.e. \eqref{Vtilde_Utility} below), with a jump at time $\tau$. Thanks to techniques of filtration expansion, we decompose $\widetilde V^{c,h}$ as a function of $\tau$ and a process $V^{c,h}$ that solves an infinite-horizon backward stochastic differential equation (BSDE) under {\it solely} the market filtration; see Proposition~\ref{prop:decompose wtV}. That is, the randomness from death and from the market can be dealt with separately. By deriving a comparison result for this infinite-horizon BSDE (Proposition~\ref{prop:comparison}), we are able to uniquely determine the Epstein-Zin utility $\widetilde V^{c,h}$ for any $k$-admissible strategy $(c,h)$ (Definition~\ref{k-admissible}); see Theorem~\ref{thm:Existence and Uniqueness}. %, the existence and uniqueness of the Epstein-Zin utility $\widetilde V^{c,h}$ follows. 
%We show the existence and uniqueness of $V^{c,h}$ whenever $(c,h)$ is $k$-admissible (see Definition~\ref{k-admissible}), and thus establish the existence and uniqueness of the Epstein-Zin utility $\widetilde V^{c,h}$. 

In a Black-Scholes financial market, we maximize the time-0 Epstein-Zin utility $\widetilde V^{c,h}_0$ over {\it permissible} strategies $(c,\pi,h)$ of consumption, investment, and healthcare spending (Definition~\ref{cP}). First, we derive the associated Hamilton-Jacobi-Bellman (HJB) equation, from which a candidate optimal strategy $(c^*,\pi^*,h^*)$ is deduced. Taking advantage of a scaling property of the HJB equation, we reduce it to a nonlinear ordinary differential equation (ODE), for which a unique classical solution exists on strength of the Perron method construction in \cite{Huang19}. This, together with a general verification theorem (Theorem~\ref{Verification}), yields the optimality of $(c^*,\pi^*,h^*)$; see Theorem~\ref{Thm:AH}. %, the major result of this paper. 

Compared with classical Epstein-Zin utility maximization, the additional controlled mortality process $M^h$ in our case adds nontrivial complexity. % to our verification argument. 
In deriving the comparison result Proposition~\ref{prop:comparison}, standard Gronwall's inequality cannot be applied due to the inclusion of $M^h$. As shown in Appendix~\ref{subsec:proof of comparison}, a transformation of processes, as well as the use of both forward and backward Gronwall's inequalities, are required to circumvent this issue. On the other hand, in carrying out verification arguments, %Particularly, for the case $\gamma>1$ and $\zeta<1$, 
we need to contain the growth of $M^h$ to ensure that the Epstein-Zin utility is well-defined. This is done through a combination of probabilistic arguments and analysis of the aforementioned nonlinear ODE; see Appendix~\ref{subsec:proof of Thm:AH} for details. 

Our model is calibrated to mortality data in the US and UK. %There are three intriguing findings. First, our model-implied mortality closely approximates actual mortality data. 
Under the simplifying assumption that the cohort born in 1900 had no healthcare and the cohort born in 1940 had full access to healthcare, we generate an endogenous mortality curve for the 1940 cohort. Figure \ref{fig:cal_US} shows that the model-implied mortality (red line) closely reproduces actual data in the US (red dots). Our model performs well also for the UK data; see Figure~\ref{fig:calibration}. We also compute the optimal healthcare spending across different ages (Figure~\ref{fig: healthcarespending}) and calibrate the efficacy of healthcare in these two countries (Figure \ref{fig:CalibratedEfficiency}).  

%Second, the calibrated efficacy of healthcare, shown in Figure \ref{fig:CalibratedEfficiency}, indicates a ranking among countries in terms of the effectiveness of healthcare spending: across realistic levels of spending, healthcare is more effective in the Netherlands than in the UK, in the UK than in the US, and in the US than in Bulgaria. 
%This ranking is in broad agreement with empirical studies on healthcare across countries; see Section~\ref{subsec:g}. Third, healthcare spendings in the four countries all increase steadily with age, %relative to total wealth and total spending, 
%but differ markedly in magnitude; see Figure~\ref{fig: healthcarespending}. This, together with the ranking of efficacy in Figure \ref{fig:CalibratedEfficiency}, reveals that higher efficacy of healthcare induces lower healthcare spending. %In other words, in the face of enhanced efficacy, our model simply stipulates less healthcare spending, instead of more to take advantage of the reduced marginal cost to curtail mortality growth.

%\red{Third, the possibility of intersection between calibrated efficacy functions may indicate differences in access to advanced medical technology vs efficiency in healthcare systems. One takeaway is that the model-generated efficacy function seems informative enough to reflect various healthcare-related realities.  }

The rest of the paper is organized as follows. Section~\ref{sec:EZ} establishes Epstein-Zin utilities over one's  random lifetime, with healthcare spending incorporated. Section~\ref{sec:PF} introduces the problem of optimal consumption, investment, and healthcare spending under Epstein-Zin preferences, and derives the related HJB equation and a general verification theorem. Section~\ref{sec:Results} characterizes optimal consumption, investment, and healthcare spending in three different settings of aging and access to healthcare. Section~\ref{sec:calibration} calibrates our model to mortality data in the US and UK. Most proofs are collected in Appendix~\ref{sec:proofs}.

 %%%%%%%%%%%%%%%%%%%%
%%%%%%%%%%%%%%%%%%%%

\section{Epstein-Zin Preferences with Healthcare Spending}\label{sec:EZ}
Let $(\Omega,\F, \P)$ be a probability space equipped with a filtration $\mathbb F= (\F_t)_{t\ge 0}$ that satisfies the usual conditions. 
%supports a standard Brownian motion $B=(B_t)_{t\ge 0}$. We denote by $\mathbb F= (\F_t)_{t\ge 0}$ the $\P$-augmentation of the natural filtration generated by $B$. 
Consider another probability space $(\Omega',\F', \P')$ supporting a random variable $Z$ that has an exponential law 
\begin{equation}\label{exp law}
\P'(Z>z)= e^{-z},\quad z\ge0. 
\end{equation}
We denote by $(\bar\Omega,\bar \F, \bar \P)$ the product probability space $(\Omega\times\Omega', \F\times\F', \P\times\P')$.  The expectations taken under $\P$, $\P'$, and $\bar \P$ will be denoted by $\E$, $\E'$, and $\bar{\E}$, respectively. 

Consider an agent who obtains utility from consumption, partially determines his lifespan through healthcare spending, and has bequest motives to leave his wealth at death to beneficiaries. 
Specifically, we assume that the mortality rate process $M$ of the agent evolves as 
\begin{equation}\label{Mortality}
dM_{t} = (\beta-g(h_{t}))M_{t}^{}dt,\quad M_{0} = m>0,
\end{equation}
where $h=(h_t)_{t\ge 0}$, a nonnegative $\mathbb F$-progressively measurable process, represents the proportion of wealth spent on healthcare at each time $t$, while $g:\R_+\to\R_+$ is the {\it efficacy function} that prescribes how much the natural growth rate of mortality $\beta>0$ is reduced by healthcare spending $h_t$.   
For any $\bar\omega = (\omega,\omega')\in\bar\Omega$, the random lifetime of the agent is formulated as
\begin{equation}\label{tau}
\tau(\bar\omega) := \inf\left\{t\ge 0 : \int_{0}^{t}M_{s}^{h}(\omega)ds \ge {Z}(\omega') \right\}.
\end{equation}
The information available to the agent is then defined as $\mathbb G = (\cG_t)_{t\ge 0}$ with
\begin{equation}\label{cG}
\cG_{t} := \F_{t}\vee\cH_t,\quad \hbox{where}\quad \cH_t := \sigma\left(\ind_{\{\tau\le u\}}, u\in[0,t]\right). 
\end{equation}
That is, at any time $t$, the agent knows the information contained in $\F_t$ and whether he is still alive (i.e. whether $\tau>t$ holds); he has no further information of $\tau$, as the random variable $Z$ is inaccessible to him. Finally, we assume that the agent can leave a fraction $\zeta\in (0,1]$, not necessarily all, of his wealth at death to beneficiaries, reflecting the effect of inheritance and estate taxes.

\begin{remark}
As a modeling simplification mainly for the purpose of tractability, the controlled mortality \eqref{Mortality} (borrowed from \cite{Huang19}) assumes that healthcare expenses relative to wealth, rather than in absolute terms, affect mortality growth. %While this is a modeling simplification, m%there are empirical and theoretical justifications; see \cite[p.319]{Huang19}. 
\end{remark}

Now, let us define a non-standard Epstein-Zin utility process that incorporates healthcare spending. First, 
%not only a consumption stream $c=(c_t)_{t\ge 0}$, a nonnegative $\mathbb F$-progressively measurable process, but also a healthcare spending process $h=(h_t)_{t\ge 0}$. 
recall the Epstein-Zin aggregator $f:\R_+\times\R\to \R$ given by
\begin{equation}\label{EZ_aggregator}
\begin{split}
f(c,v)&:=\delta\frac{(1-\gamma)v}{1-\frac{1}{\psi}}\left(\bigg(\frac{c}{((1-\gamma)v)^{\frac{1}{1-\gamma}}}\bigg)^{1-\frac{1}{\psi}}-1\right)\\
&=\delta\frac{c^{1-\frac{1}{\psi}}}{1-\frac{1}{\psi}}\big((1-\gamma)v\big)^{1-\frac{1}{\theta}}-\delta\theta v,\qquad\text{with}\quad \theta := \frac{1-\gamma}{1-\frac{1}{\psi}},
\end{split}
\end{equation}
where $\gamma$ and $\psi$ represent the agent's {relative risk aversion} and EIS, respectively, as stated in Section~\ref{sec:intro}. Given a consumption stream $c=(c_t)_{t\ge 0}$, assumed to be nonnegative $\mathbb F$-progressively measurable, and a healthcare spending process $h=(h_t)_{t\ge 0}$ introduced below \eqref{Mortality}, we define the {\it Epstein-Zin utility on the random horizon $\tau$} to be a $\mathbb G$-adapted semimartingale %c\`{a}dl\`{a}g process 
$(\widetilde{V}_{t}^{c,h})_{t\ge 0}$ satisfying %(i) it is discontinuous only at $t=\tau$, and (ii) it follows the dynamics
\begin{equation}\label{Vtilde_Utility}
\widetilde{V}_{t}^{c,h} = \bar{\E}_t\left[\int_{t\wedge\tau}^{T\wedge\tau}f(c_s,\widetilde{V}_{s}^{c,h})ds + \zeta^{1-\gamma}\widetilde{V}_{\tau-}^{c,h}\ind_{\{\tau\le T\}} +  \widetilde{V}_{T}^{c,h}\ind_{\{\tau>T\}} \right],\text{ for all }0\le t\le T<\infty,
\end{equation}
where we use the notation $\bar \E_{t}\left[\cdot\right]=\bar\E\left[\cdot|\cG_{t}\right]$. 
In \eqref{Vtilde_Utility}, we assert that the loss of wealth at death results in a decreased bequest utility, by a factor of $\zeta^{1-\gamma}$. This assertion will be made clear and justified in Section~\ref{sec:Results}, where a financial model is in place; see Remark~\ref{rem:confirms} particularly. 

Before solving \eqref{Vtilde_Utility} for $(\widetilde V^{c,h}_t)_{t\ge 0}$, we introduce a general definition of infinite-horizon BSDEs.% that will be used throughout the paper. 

\begin{definition}\label{def: BSDE}
Let $V$ be an $\mathbb F$-progressively measurable process satisfying $\E[\sup_{s\in[0,t]}|V_{s}|]<\infty$ for all $t\ge 0$. For any 
$G:\Omega\times\R_{+}\times\R\to \R$, we say $V$ is a solution to the infinite-horizon BSDE  %generated by $G$, i.e.	
\begin{align}\label{BSDE}
	d V_{t}= -G(\omega,t,V_t) dt + d\mathscr{M}_{t},%\quad\forall t\ge 0.
\end{align}
if the following conditions hold: (i) $\big(G(\cdot,t,V_t(\cdot))\big)_{t\ge 0}$ is $\mathbb F$-progressively measurable, and (ii) for any $T>0$ there exists an $\mathbb{F}$-martingale $(\mathscr{M}_t)_{t\in[0,T]}$ such that \eqref{BSDE} holds for $0\le t\le T$. 
\end{definition} 

\begin{remark}
Without a terminal condition, \eqref{BSDE} can have infinitely many solutions. Indeed, as long as $G$ admits proper monotonicity, there are solutions to \eqref{BSDE} that satisfy ``$\lim_{t\to\infty}V_{t} = \xi$ for $\F$-measurable random variable $\xi$'' or ``$\lim_{t\to\infty}\E\left[e^{\rho t}V_{t}\right]\rightarrow 0$ for $\rho>0$''; see \cite{Briand03} and \cite{Fan15}. %)depending on the constant of monotonicity. 
%Since there is no natural terminal condition in the context of Epstein-Zin utilities, we leave it open and 
We will address this non-uniqueness issue by enforcing appropriate ``terminal behavior''; see Remark~\ref{rem:address}.
\end{remark}
%This implicitly enforces an exponential decay rate in $V$ and forces $\xi = 0$. However, under such conditions, the corresponding solution may easily degenerate to something constant or trivial. 

The next result shows that the $\mathbb G$-adapted $\widetilde V$ in \eqref{Vtilde_Utility} can be expressed as a function of $\tau$ and an $\mathbb F$-adapted process $V$ that satisfies an infinite-horizon BSDE. 

\begin{proposition}\label{prop:decompose wtV}
Let $c, h$ be nonnegative $\mathbb F$-progressively measurable.  Then, $\widetilde V$ is a $\mathbb G$-adapted semimartingale, with $\bar \E[\sup_{s\in[0,t]}|\widetilde V_{s}|]<\infty$ for all $t\ge 0$, that satisfies \eqref{Vtilde_Utility} if and only if 
\begin{equation}\label{wtV decompose}
\widetilde V_t = V_t \ind_{\{t<\tau\}} + \zeta^{1-\gamma}V_{\tau-}\ind_{\{t\ge\tau\}}\qquad \forall t\ge 0,
\end{equation}
where $V$ is an $\mathbb F$-adapted semimartingale, with $\E[\sup_{s\in[0,t]}|V_{s}|]<\infty$ for all $t\ge 0$, that satisfies the infinite-horizon BSDE 
\begin{equation}\label{BSDE_V}
	dV_{t}^{} = -F(c_{t},M_{t}^{h},V_{t}^{})ds + d\mathscr{M}_{t},%\quad\text{for all}\ 0\le t\le T<\infty,
\end{equation}
with $F:\R_+\times\R_+\times\R\to\R$ defined by
\begin{equation}\label{F func}
F(c,m,v):= f(c,v) - (1-\zeta^{1-\gamma})m v .
\end{equation}
\end{proposition}

\begin{proof}
See Section~\ref{subsec:proof of prop:decompose wtV}.
\end{proof}

\begin{remark}
Proposition~\ref{prop:decompose wtV} actually holds more generally beyond the specific driver $f$ in \eqref{EZ_aggregator} and the boundary condition ``$\widetilde V_T = \zeta^{1-\gamma} \widetilde V_{\tau-}$'' for $T\ge \tau$ encoded in \eqref{Vtilde_Utility}. Specifically, in \eqref{Vtilde_Utility}, if we allow for a general Borel driver $f$ and replace $\zeta^{1-\gamma} \widetilde V_{\tau-}$ therein by $H(\widetilde V_{\tau-})$ for a general continuous function $H$ that grows at most linearly, the arguments in the proof of Proposition~\ref{prop:decompose wtV} (see Section~\ref{subsec:proof of prop:decompose wtV}) still hold, leading to an upgraded version of Proposition~\ref{prop:decompose wtV} with $\zeta^{1-\gamma}V_{\tau-}$ in \eqref{wtV decompose} and $\zeta^{1-\gamma}v$ in \eqref{F func} replaced by $H(V_{\tau-})$ and $H(v)$, respectively. 
\end{remark}

In view of Proposition~\ref{prop:decompose wtV}, to uniquely determine the Epstein-Zin utility process $\widetilde V$, we need to find a suitable class of stochastic processes among which there exists a unique solution to \eqref{BSDE_V}. To this end, we start with imposing appropriate integrability and transversality conditions.  %focus on a specific collection of $(c,h)$. 

%\begin{remark}\label{Homothetic}
%	As required above, solutions to this BSDE are homothetic. Briefly, consider an ``admissible" $(c,h)$, a strategy for which $V_{t}^{c,h}$ satisfies appropriate regularity conditions defined in \ref{TransversalityCondition}. Then, notice the mapping $(c,h)\mapsto V_{t}^{(c,h)}$ is homothetic. To see this, denote $\lambda^{1-\gamma}V_{t}^{c,h}:= V_{t}^{\zeta, c, h}$ and take advantage of $\lambda^{1-\gamma}f(c,V_{s}) = f(\lambda c, \lambda^{1-\gamma}V_{s})$ to find:
	%\begin{align*}
	%V_{t}^{\lambda,c,h} & = V_{T}^{\lambda,c,h} + \int_{t}^{T}\bigg(f(\lambda c_{s},V_{t}^{\lambda,c,h} ) - (1-\zeta^{1-\gamma})M_{s}^{h}V_{t}^{\lambda,c,h})\bigg)ds - \int_{t}^{T}\lambda^{1-\gamma}Z_{s}dB_{s}.
	%\end{align*}
	%However, since any admissible strategy yields unique solutions, we must have $V_{t}^{\lambda,c,h} = V_{t}^{\lambda c,h}$. As such, $\lambda^{1-\gamma}V_{t}^{c,h} = V_{t}^{\lambda c,h}$ a.s. for $\lambda>0$.
%\end{remark}

\begin{definition}\label{cE}
For any $k\in\R$, define $\Lambda:= \delta\theta + (1-\theta)k$. Then, for any nonnegative $\mathbb F$-progressively measurable $h$, we denote by $\cE^h_k$ the set of all $\mathbb F$-adapted semimartingales $Y$ that satisfy the following integrability and transversality conditions:
\begin{align}\label{TransversalityCondition}
	\E\bigg[\sup\limits_{s\in[0,t]}\left|Y_{s}\right|\bigg]<\infty\ \ \forall t>0\quad\hbox{and}\quad\lim\limits_{t\rightarrow\infty}e^{-\Lambda t}\E\left[e^{-\gamma(\psi-1)\frac{1-\zeta^{1-\gamma}}{1-\gamma}\int_{0}^{t}M_{s}^{h}ds}|Y_{t}|\right] = 0.
	\end{align}
\end{definition}

\begin{remark}
Condition \eqref{TransversalityCondition} is similar to \cite[(2.3)]{Melnyk17}, but the controlled mortality $M^h$ in our case complicates the transversality condition: unlike \cite[(2.3)]{Melnyk17}, the exponential term no longer contains a constant rate, but a stochastic one involving $M^h$. This adds nontrivial complexity to deriving a comparison result (Proposition~\ref{prop:comparison}) and the use of verification arguments (Theorem~\ref{Thm:AH}).  
\end{remark}

\begin{remark}\label{rem:negative Lambda}
The constant $\Lambda:= \delta\theta + (1-\theta)k$ in \eqref{TransversalityCondition} can be negative, even when $k>0$ (as will be assumed in Section~\ref{sec:Results}). In such a case, \eqref{TransversalityCondition} stipulates that $M^h$ must increase fast enough to neutralize the growth of $e^{-\Lambda t}$, such that the transversality condition can be satisfied. 
\end{remark}

We now introduce the appropriate collection of strategies $(c,h)$ we will focus on. 

\begin{definition}\label{k-admissible}
	Let $c, h$ be nonnegative $\mathbb F$-progressively measurable. For any $k\in\R$, we say $(c,h)$ is $k$-admissible if there exists $V\in \cE^h_k$ satisfying \eqref{BSDE_V} and 
	\begin{equation}\label{comparison_condition}
	V_{s}\le \delta^{\theta}\left(k+(\psi-1)\frac{1-\zeta^{1-\gamma}}{1-\gamma} M^h_{s}\right)^{-\theta}\frac{c_{s}^{1-\gamma}}{1-\gamma},\quad \forall s\ge 0.
	\end{equation}
	%We denote by $\A_{k}$ the collection of all $k$-admissible strategy $(c,h)$. 
\end{definition}

\begin{remark}\label{rem:address}
Condition \eqref{comparison_condition} is the key to a comparison result for \eqref{BSDE_V}, as shown in Proposition~\ref{prop:comparison} below.
%so that the Epstein-Zin utility $\widetilde V$ can be uniquely identified via \eqref{wtV decompose}. 
%The nuance here, also noted in Appendix C of \cite{DuffieEpstein92}, Example 1 and 2 of \cite{Royer04}, and circumvented in Theorem 14.10 of \cite{Hu11}, is in replacing the role of a ``\textit{terminal condition}" with a ``\textit{terminal behavior}". 
In a sense, \eqref{TransversalityCondition}-\eqref{comparison_condition} is the enforced ``\textit{terminal behavior}", under which a solution to \eqref{BSDE} can be uniquely identified. 
Technically, \eqref{comparison_condition} is similar to typical conditions imposed for infinite-horizon BSDEs, such as \cite[(H1')]{Briand03} and the one in \cite[Theorem 5.1]{Fan15}: all of them require the solution to be bounded from above by a tractable process. % I.e., we both require $V_{t}$ to satisfy a growth condition, ours is just explicitly detailed under a specific generator $G$.
Moreover, for classical Epstein-Zin utilities (without healthcare), a similar condition was imposed in \cite[(2.5)]{Melnyk17}. 
In fact, Definition~\ref{k-admissible} is in line with \cite[Definition 2.1]{Melnyk17}, but adapted to include the controlled mortality $M^h$. 
\end{remark}

A comparison result for BSDE \eqref{BSDE_V} can now be established.

\begin{proposition}\label{prop:comparison}
	Let $k\in\R$ and $c, h$ be nonnegative $\mathbb F$-progressively measurable processes. %, and $G:\Omega\times\R_{+}\times\R\mapsto\R$ an $\F\otimes\B(\R_{+})$ measurable mapping. 
	Suppose that $V^1\in\cE_{k}^{h}$ is a solution to \eqref{BSDE_V} and $V^2\in\cE_{k}^{h}$ is a solution to \eqref{BSDE}. 
	%\begin{enumerate}
		%\item [(i)] 
		If $V^{1}$ satisfies \eqref{comparison_condition} and $F(c_{t},M_{t},V_{t}^2) \le G(t,V_{t}^2)\ d\P\times dt$-a.e., then $V_{t}^{1}\le V_{t}^{2}$ for $t\ge 0$ $\P$-a.s.
		%\item [(ii)] If $V^{2}$ satisfies \eqref{comparison_condition} and $G(t,V_{t}^2) \le F(c_{t},M_{t},V_{t}^2)\ d\P\times dt$-a.e., then $V_{t}^{1}\ge V_{t}^{2}$ for $t\ge 0$ $\P$-a.s.
	%\end{enumerate}
\end{proposition}

\begin{proof}
See Section~\ref{subsec:proof of comparison}.
\end{proof}

The next result is a direct consequence of Propositions~\ref{prop:decompose wtV} and \ref{prop:comparison}. 

\begin{theorem}\label{thm:Existence and Uniqueness}
Fix $k\in\R$. For any $k$-adimissible $(c,h)$, there exists a unique solution $V^{c,h}\in\cE^h_k$ to \eqref{BSDE_V} that satisfies \eqref{comparison_condition}. Hence, the Epstein-Zin utility $\widetilde{V}^{c,h}$ can be uniquely determined via \eqref{wtV decompose}.
\end{theorem}

\begin{remark}
Results in this section in fact hold true more generally, for certain specifications of $(\psi, \gamma)$ that do not fulfill Assumption~\ref{specification}. For instance, for the cases ``$\psi\in(0,1)$, $\gamma>\frac{1}{\psi}$'' and ``$\psi, \gamma\in(0,1)$'',  Theorem~\ref{thm:Existence and Uniqueness} can be similarly established by suitably adjusting \eqref{TransversalityCondition}-\eqref{comparison_condition}. This will allow the main result of this paper, Theorem~\ref{Thm:AH} below, to be generalized to these cases. 
\end{remark}

\section{Problem Formulation}\label{sec:PF}
Let $B=(B_t)_{t\ge 0}$ be an $\mathbb F$-adapted standard Brownian motion. Consider a financial market with a riskfree rate $r>0$ and a risky asset $S_{t}$ given by
\begin{align}\label{S}
dS_{t} & = (\mu+r)S_{t}dt + \sigma S_{t}dB_{t},
\end{align}
where $\mu\in\R$ and $\sigma>0$ are given constants. Given initial wealth $x>0$, at each time $t\ge0$, an agent chooses a consumption rate $c_{t}$, invests a fraction $\pi_t$ of his wealth on the risky asset, and spends another fraction $h_{t}$ on healthcare. The resulting dynamics of the wealth process $X$ is
\begin{align}\label{Wealth}
dX_{t} &= X_{t}\left(r+\mu\pi_{t}-h_{t}\right)dt -c_{t} dt + X_{t}\sigma\pi_{t}dB_{t},\quad X_0 =x.
\end{align}

\begin{definition}\label{H_k}
For all $k\in\R$, let $\cH_k$ be the set of strategies $(c,\pi,h)$ such that $(c,h)$ is $k$-admissible (Definition~\ref{k-admissible}), $\pi$ is $\mathbb F$-progressively measurable, and a unique solution $X^{c,\pi,h}$ to \eqref{Wealth} exists.
\end{definition}

The agent aims at maximizing his lifetime Epstein-Zin utility $\widetilde V_0^{c,h}$ by choosing $(c,\pi,h)$ in a suitable collection of strategies $\cP$, i.e.
\begin{align}\label{problem}
%V(x,m) := 
\sup\limits_{(c,\pi,h)\in\cP}\widetilde{V}_{0}^{c,h} &=\sup\limits_{(c,\pi,h)\in\cP}V_{0}^{c,h},
\end{align}
where the equality follows from \eqref{wtV decompose}. In this section, we only require $\cP$ to satisfy
\begin{equation}\label{cP condition}
\cP\subseteq \cH_k% \{(c,\pi,h) : (c,h)\in\A_k,\ \pi\ \hbox{is $\mathbb F$-progressively measurable}\}
\quad \hbox{for some}\ k\in\R.
\end{equation}
Our focus is to %derive the Hamilton-Jacobi-Bellman (HJB) equation for \eqref{problem} and 
establish a versatile verification theorem under merely %the general condition 
\eqref{cP condition}. 
A more precise definition of $\cP$, depending on specification of $\beta$, $\gamma$, and $\zeta$, will be introduced in Definition~\ref{cP}.

%%%%%%%%%%%%%%%%%%%%%%%%%%

\subsection{A General Verification Theorem}\label{DPP}
Under the current Markovian setting (i.e. \eqref{S} and \eqref{Wealth}), we take
\begin{equation}\label{problem'}
v(x,m) : = \sup\limits_{(c,\pi,h)\in\cP}V_{0}^{c,h},
\end{equation}
i.e. the optimal value should be a function of the current wealth and mortality. The relation \eqref{V eqn}, derived from \eqref{Vtilde_Utility}, suggests the following dynamic programming principle: With the shorthand notation $\fp = (c,\pi,h)$ and $\fp_s = (c_s,\pi_s,h_s)$ for $s\ge0$, for any $T>0$, 
\begin{align}\label{ValueFn}
&v(x,m) =\nonumber\\ &\sup\limits_{\fp\in\cP}\E\left[\int_{0}^{T}e^{-\int_{0}^{s}M_{r}^{h}dr}\left(f(c_{s},v(X_{s}^{\fp},M_{s}^{h})) + \zeta^{1-\gamma}M_{s}^{h}v(X_{s}^{\fp},M_{s}^{h})\right)ds + e^{-\int_{0}^{T}M_{s}ds}v(X_{T}^{\fp},M_{T}^{h}) \right].
\end{align}
By applying It\^{o}'s formula to $e^{-\int_{0}^{t}M_{s}^{h}ds}v(X_{t}^{\fp},M_{t}^{h})$, assuming enough regularity of $v$, we get
\begin{align*}
&e^{-\int_{0}^{T}M_{s}^{h}ds}v(X_{T}^{\fp},M_{T}^{h}) - v(x,m)\\
&= \int_{0}^{T}\left(L^{\fp_s}[v](X_{t}^{\fp},M_{t}^{h}) - M_{t}^{h}v(X_{t}^{\fp},M_{t}^{h})\right)dt + \int_{0}^{T}e^{-\int_{0}^{t}M_{s}^{h}ds}\sigma\pi X_{t}^{\fp}v_{x}(X_{t}^{\fp},M_{t}^{h})dB_{t},
\end{align*}
where the operator $L^{a,b,d}[\cdot]$ is defined by 
\begin{equation}\label{L}
L^{a,b,d}[\kappa](x,m):= \left((r+\mu b-d)x-a\right) \kappa_{x}(x,m)+(\beta-g(d))m \kappa_{m}(x,m) + \frac{1}{2}\sigma^{2}b^{2}x^{2}\kappa_{xx}(x,m),
\end{equation}
for any $\kappa\in C^{2,1}(\R_+\times\R_+)$. We can then rewrite \eqref{ValueFn} as
\begin{equation*}
0 = \sup\limits_{\fp\in\cP}\E\left[\int_{0}^{T}e^{-\int_{0}^{s}M_{t}^{h}dt}\left(f(c_{s},v(X_{s}^{\fp},M_{s}^{h})) + (\zeta^{1-\gamma}-1)M_{s}^{h} v(X_{s}^{\fp},M_{s}^{h})+L^{\fp_s}[v](X_{s}^{\fp},M_{s}^{h})\right)ds\right].
\end{equation*}
The HJB equation associated with $v(x,m)$ is then
\begin{align}\label{HJB}
0 &= \sup\limits_{c\in\R_{+}}\left\{f(c,w(x,m)) - cw_{x}(x,m)\right\} +\sup\limits_{h\in\R_{+}}\left\{-g(h)mw_{m}(x,m) -hxw_{x}(x,m)\right\}\nonumber\\
&\hspace{0.2in}+ \sup\limits_{\pi\in\R}\left\{\mu\pi xw_{x}(x,m) + \frac{1}{2}\sigma^{2}\pi^{2}x^{2}w_{xx}(x,m)\right\}\\
&\hspace{0.2in}+ rxw_{x}(x,m) +\beta mw_{m}(x,m) + (\zeta^{1-\gamma}-1)m w(x,m),\quad \forall (x,m)\in\R_{+}^2.\nonumber
\end{align}
Equivalently,  this can be written in the more compact form
\begin{equation}\label{HJB'}
\sup\limits_{c, h\in\R_+, \pi\in\R}\left\{L^{c,\pi,h}[w](x,m) + f(c,w(x,m))\right\} + (\zeta^{1-\gamma}-1)m w(x,m)= 0,\quad \forall (x,m)\in\R_{+}^2.
\end{equation}

%We establish a general verification theorem for $v(x,m)$ in \eqref{problem'}.

\begin{theorem}\label{Verification}
Let $w\in C^{2,1}(\R_{+}\times\R_+)$ be a solution to \eqref{HJB} and $\cP$ satisfy \eqref{cP condition}. Suppose for any $(c,\pi,h)\in\cP$, the process $w(X_{t}^{c,\pi,h},M_{t}^{h})$, $t\ge 0$, belongs to $\cE^h_k$ (with $k\in\R$ specified by \eqref{cP condition}) and
\begin{equation}\label{veri condition}
\E\bigg[\sup_{s\in[0,t]} \pi_s X^{c,\pi,h}_s w_x(X_{s}^{c,\pi,h},M_{s}^{h})\bigg]<\infty,\quad \forall t>0.
\end{equation}
 Then, the following holds. 
\begin{itemize}
\item [(i)] $w(x,m)\ge v(x,m)$ on $\R_{+}\times \R_+$.
\item [(ii)] Suppose further that there exist Borel measurable functions $\bar c, \bar\pi, \bar h:\R^2_+\to\R$ such that $\bar c(x, m)$, $\bar\pi(x,m)$, and $\bar h(x, m)$ are maximizers of
\begin{align}
&\sup\limits_{c\in\R_{+}}\left\{f(c,w(x,m)) - cw_{x}(x,m)\right\},\quad \sup\limits_{\pi\in\R}\left\{\mu\pi xw_{x}(x,m) + \frac{1}{2}\sigma^{2}\pi^{2}x^{2}w_{xx}(x,m)\right\},\label{sup1, 2}\\
&\hspace{1.2in} \sup\limits_{h\in\R_{+}}\left\{-g(h)mw_{m}(x,m) -hxw_{x}(x,m)\right\},\label{sup3}
\end{align}
respectively, for all $(x,m)\in\R^2_+$. If $(c^*,\pi^*,h^*)$ defined by
\begin{equation}\label{optimal strategies}
c^*_t:= \bar c(X_t,M_t),\quad \pi^*_t:= \bar \pi(X_t,M_t),\quad h^*_t:=\bar h(X_t,M_t),\qquad t\ge 0,
\end{equation}
belongs to $\cP$ and $W^*_t:=w(X_{t}^{c^*,\pi^*,h^*},M_{t}^{h^*})$ satisfies \eqref{comparison_condition} (with $V$, $c$, $h$ replaced by $W^*$, $c^*$, $h^*$), then $(c^*,\pi^*,h^*)$ optimizes \eqref{problem'} and $w(x,m)=v(x,m)$ on $\R_{+}\times \R_+$.
\end{itemize}
\end{theorem}

\begin{proof}
(i) Fix $(x,m)\in\R^2_+$. Consider an arbitrary $\fp=(c,\pi,h)\in\cP$. For any $T\ge0$ and $t\in[0,T]$, by applying It\^{o}'s formula to $w(X^{\fp}_s,M^h_s)$, we get %thanks to the assumed regularity of $w$, gives
	\begin{align*}
	 w(X_{T}^{\fp},M_{T}^{h}) = w(X_{t}^{\fp},M_{t}^{h})+\int_{t}^{T}L^{\fp_{s}}[w](X_{s}^{\fp},M_{s}^{h})ds +\int_{t}^{T}\sigma\pi_{s}X_{s}^{\fp}w_{x}(X_{s}^{\fp},M_{s}^{h})dB_{s}, 
	\end{align*}
	where the operator $L^{a,b,d}[\cdot]$ is defined in \eqref{L}. Thanks to \eqref{veri condition}, $u\mapsto \int_{t}^{u}\sigma\pi_{s}X_{s}^{\fp}w_{x}(X_{s}^{\fp},M_{s}^{h})dB_{s}$ is a true martingale. Hence, the above equality shows that $W_s:= w(X_{s}^{\fp},M_{s}^{h})$ is a solution to BSDE \eqref{BSDE}, with $G(\omega, s, v) := -L^{\fp_{s}(\omega)}[w](X_{s}^{\fp}(\omega),M_{s}^{h}(\omega))$. On the other hand, \eqref{cP condition} implies that $(c,h)$ is $k$-admissible, so that there exists a unique solution $V^{c,h}\in\cE^h_k$ to \eqref{BSDE_V} that satisfies \eqref{comparison_condition} (Theorem~\ref{thm:Existence and Uniqueness}). Since $w$ is a solution to \eqref{HJB}, and equivalently to \eqref{HJB'}, we have
	\begin{equation}\label{verif inequality}
	F(c_s,M^h_s, W_s) = f(c_{s},W_{s}) + (\zeta^{1-\gamma}-1)M_{s}^{h}W_{s} \le -L^{\fp_{s}}[w](X_{s}^{\fp},M_{s}^{h}).
	\end{equation}
	 We then conclude from Proposition~\ref{prop:comparison} that $W_{t}\ge V_{t}^{c,h}$ for all $t\ge0$. In particular, $w(x,m)=W_0\ge V_0^{c,h}$. By the arbitrariness of $(c,\pi,h)\in\cP$, $w(x,m)\ge \sup_{(c,\pi,h)\in\cP}V_0^{c,h}=v(x,m)$, as desired.

(ii) Fix $(x,m)\in\R^2_+$. If $(c^*,\pi^*,h^*)\in\cP$, we can repeat the arguments in part (a), obtaining \eqref{verif inequality} with the inequality replaced by equality. This shows that $W^*_t=w(X^{c^*,\pi^*,h^*}_t,M^{h^*}_t)\in \cE^{h^*}_k$ is a solution to \eqref{BSDE_V}. Also, \eqref{cP condition} implies that $(c^*,h^*)$ is $k$-admissible, so that there is a unique solution $V^{c^*,h^*}\in\cE^{h^*}_k$ to \eqref{BSDE_V} satisfying \eqref{comparison_condition} (Theorem~\ref{thm:Existence and Uniqueness}). As $W^*$ also satisfies \eqref{comparison_condition}, we have $W^*_t= V^{c^*,h^*}_t$ for all $t\ge 0$; particularly, $w(x,m)= W^*_0 = V^{c^*,h^*}_0$. With $w(x,m)\ge \sup_{(c,\pi,h)\in\cP}V_0^{c,h}=v(x,m)$ in part (a), we conclude $w(x,m)=v(x,m)$ and $(c^*,\pi^*,h^*)\in\cP$ is an optimal control. 
\end{proof}

%%%%%%%%%%%%%%%%%%%%%%%%%

\subsection{Reduction to an Ordinary Differential Equation}
If we assume heuristically that $w_{xx}<0$, $w_{m}<0$, $g$ is differentiable, and the inverse of $g'$ is well-defined, then the optimizers stated in Theorem~\ref{Verification} (ii) can be uniquely determined as  %(compare with the policies in equation (4.8) of \cite{Huang19})
\begin{equation}\label{candidates}
\begin{split}
\bar c(x,m) &= \delta^{\psi}\frac{\left[(1-\gamma)w(x,m)\right]^{\psi(1-\frac{1}{\theta})}}{w_{x}(x,m)^{\psi}},\quad \bar\pi(x,m) = - \frac{\mu}{\sigma^{2}}\frac{w_{x}(x,m)}{xw_{xx}(x,m)},\\
 &\hspace{0.5in}\bar h(x,m) = (g')^{-1}\left(-\frac{xw_{x}(x,m)}{mw_{m}(x,m)}\right).
\end{split}
\end{equation}
Plugging these into \eqref{HJB} yields
\begin{align}\label{HJB_S}
0 &= \frac{\delta^{\psi}}{\psi-1}\frac{\left[(1-\gamma)v(x,m)\right]^{\psi(1-\frac{1}{\theta})}}{v_{x}(x,m)^{\psi-1}} - \delta\theta v(x,m)- \frac{1}{2}\left(\frac{\mu}{\sigma}\right)^{2}\frac{v_{x}(x,m)^{2}}{v_{xx}(x,m)} + rxv_{x}(x,m) +\beta mv_{m}(x,m)\nonumber\\
&\hspace{0.2in} + (\zeta^{1-\gamma}-1)mv(x,m)- mv_{m}(x,m)\sup\limits_{h\in\R_{+}}\left\{g(h) +\frac{hxv_{x}(x,m)}{mv_{m}(x,m)}\right\}.
\end{align}
Using the ansatz $w(x,m) = \delta^{\theta}\frac{x^{1-\gamma}}{1-\gamma}u(m)^{-\frac{\theta}{\psi}}$, the above equation reduces to 
%\begin{align}
%0 &= \frac{\delta^{\theta}}{\psi-1}x^{1-\gamma}u(m)^{1-\frac{\theta}{\psi}} - \frac{\delta^{1+\theta}}{1-\frac{1}{\psi}}x^{1-\gamma}u(m)^{-\frac{\theta}{\psi}}+ \frac{1}{2\gamma}\delta^{\theta}\left(\frac{\mu}{\sigma}\right)^{2}x^{1-\gamma}u(m)^{-\frac{\theta}{\psi}} + r\delta^{\theta}x^{1-\gamma}u(m)^{-\frac{\theta}{\psi}}\nonumber\\
%&+\frac{\beta m}{1-\psi}\delta^{\theta}x^{1-\gamma}u(m)^{-\frac{\theta}{\psi}-1}u'(m) + m(\zeta^{1-\gamma}-1)\delta^{\theta}\frac{x^{1-\gamma}}{1-\gamma}u(m)^{-\frac{\theta}{\psi}}\nonumber\\
%&- \frac{ m}{1-\psi}\delta^{\theta}x^{1-\gamma}u(m)^{-\frac{\theta}{\psi}-1}u'(m)\sup\limits_{h\in\R_{+}}\left\{g(h) -(\psi-1)\frac{hu(m)}{ mu'(m)}\right\}.
%\end{align}
%This can be easily reduced by pre-multiplying both sides by $\frac{(1-\gamma)\psi}{\theta}\delta^{-\theta}x^{\gamma-1}u^{1+\frac{\theta}{\psi}}$ giving:
\begin{align}\label{ODE}
0 &= u(m)^{2}-\tilde{c}_{0}(m) u(m)-\beta mu'(m) + mu'(m)\sup\limits_{h\in\R_{+}}\left\{g(h) -(\psi-1)\frac{u(m)}{ mu'(m)}h\right\},\quad m>0,
\end{align}
where 
\begin{align}\label{c_0}
\tilde{c}_{0}(m) &:= \psi\delta + (1-\psi)\bigg(\frac{(\zeta^{1-\gamma}-1)m}{1-\gamma}+r+\frac{1}{2\gamma}\left(\frac{\mu}{\sigma}\right)^{2}\bigg).
%\psi\left(\delta-\frac{1}{\theta}(\zeta^{1-\gamma}-1)m\right) + (1-\psi)\bigg(r+\frac{1}{2\gamma}\left(\frac{\mu}{\sigma}\right)^{2}\bigg).%,\\
%&= k^* -\frac{\psi}{\theta}(\zeta^{1-\gamma}-1)m,\qquad \hbox{with $k^*$ defined as in \eqref{k^*}}. \nonumber
\end{align}
Moreover, the maximizers in \eqref{candidates} now become
\begin{equation}\label{bar's}
\bar c(x,m) = x u(m),\quad \bar \pi \equiv \frac{\mu}{\gamma\sigma^{2}}, \quad \bar h(m) = (g')^{-1}\left((\psi-1)\frac{u(m)}{mu'(m)}\right).
\end{equation}
%In Section~\ref{sec:Results}, 
These maximizers indeed characterize optimal consumption, investment, and healthcare spending, as will be shown in the next section. % among an appropriate class of strategies.

%%%%%%%%%%%%%%%%%%%%%%%%%%%
%%%%%%%%%%%%%%%%%%%%%%%%%%%

\section{The Main Results}\label{sec:Results}
Let us now formulate the set $\mathcal P$ of  {\it permissible} strategies $(c,\pi,h)$ in the optimization problem \eqref{problem}. First, take $k\in\R$ in Definition \ref{cE} to be 
\begin{equation}\label{k^*}
k^*:= \delta\psi + (1-\psi)\bigg(r+\frac{1}{2\gamma}\left(\frac{\mu}{\sigma}\right)^{2}\bigg),
\end{equation}
so that $\Lambda\in\R$ in Definition \ref{cE} becomes
\begin{align}\label{Lambda^*}
\Lambda^* := \delta\theta + (1-\theta)k^* =\delta\gamma\psi + (1-\gamma\psi)\left(r+\frac{1}{2\gamma}\left(\frac{\mu}{\sigma}\right)^{2}\right).
\end{align} 

\begin{definition}
Let $\cP_1$ the set of strategies $(c,\pi,h)$ such that $(c,\pi,h)\in \cH_{k^*}$, $(X^{c,\pi,h})^{1-\gamma}$ satisfies \eqref{TransversalityCondition} (with $\Lambda\in\R$ therein taken to be $\Lambda^*$) as well as
%\begin{equation}
$\E\big[\sup_{s\in[0,t]} \pi_s(X^{c,\pi,h}_s)^{1-\gamma}\big]<\infty$ for $t\ge 0.$
%\end{equation}

Let $\cP_2$ be defined as $\cP_1$, except that the second part of \eqref{TransversalityCondition} is replaced by
\begin{equation}\label{Permissible}
\lim_{t\rightarrow\infty}e^{-\Lambda^* t}\E\left[e^{-\eta\gamma(\psi-1)\frac{1-\zeta^{1-\gamma}}{1-\gamma}\int_{0}^{t}M_{s}^{h}ds}(X_{t}^{c,\pi,h})^{1-\gamma}\right] = 0,\quad \hbox{for some $\eta\in(1-\frac{1}{\gamma},1)$}. 
\end{equation}
\end{definition}

\begin{definition}\label{cP}
The set of {\it permissible} strategies $(c,\pi,h)$, denoted by $\mathcal P$, is defined as follows.
\begin{itemize}
\item [(i)] For the case $\beta=0$ and $g\equiv 0$ (i.e. with neither aging nor healthcare), $\cP:=\cP_1$;
\item [(ii)] For the case $\beta>0$ (i.e. with aging), 
\[
\cP:=
\begin{cases}
\cP_1,\quad &\hbox{if}\ \gamma\in\big(\frac{1}{\psi},1\big)\ \hbox{or}\ \zeta= 1,\\
\cP_2,\quad &\hbox{if}\ \gamma> 1\ \hbox{and}\ \zeta\in(0,1),
\end{cases}
\]
\end{itemize}
\end{definition}

\begin{remark}
When there is aging ($\beta>0$), for the case $\gamma> 1$ and $\zeta\in(0,1)$, we need $(X^{c,\pi,h})^{1-\gamma}$ to satisfy the slightly stronger condition \eqref{Permissible} (than the transversality condition in \eqref{TransversalityCondition}), so that the general verification Theorem~\ref{Verification} can be applied; see Appendix~\ref{subsec:proof of Thm:AH} for details. 
\end{remark}

The rest of the section presents main results in three different settings of aging and access to healthcare, in order of complexity. %progressively: Section~\ref{beta=0, g=0} deals with the simplest case with neither aging nor healthcare; Section~\ref{beta>0, g=0} handles the scenario with aging but without healthcare, which serves as a baseline for the general model with both aging and healthcare in Section~\ref{beta>0, g>0}.
%%%%%%%%%%%%%%%%%%%%%%%

\subsection{Neither Aging nor Healthcare}\label{beta=0, g=0}
When the natural growth rate of mortality is zero ($\beta=0$) and healthcare is unavailable ($g\equiv 0$), the mortality process is  constant, i.e. $M_{t}\equiv m$. Consequently, in the HJB equation \eqref{HJB}, all derivatives in $m$ should vanish; also, as $v(x,m)$ is nondecreasing in $x$ by definition, the second supremum in \eqref{HJB} should be zero. Corresponding to this largely simplified HJB equation, \eqref{ODE} reduces to 
\[
0 = u(m)^{2}-\tilde{c}_{0}(m) u(m), 
\]
which directly implies $u(m)=\tilde c_0(m)$. 
%One chooses $h_t\equiv 0$ in this setting, as positive healthcare spending yields no effect. Also, Definition \ref{cE} no longer depends on the healthcare spending process, and Definition~\ref{k-admissible} becomes a condition of solely the consumption stream $c$. 
The problem \eqref{problem'} can then be solved explicitly. %Recall $\tilde c_0(m)$ in \eqref{c_0}. 

\begin{proposition}\label{prop:NoAging}
Assume $\beta=0$ and $g\equiv 0$. For any $m\ge 0$, if $\tilde{c}_{0}(m)>0$ in \eqref{c_0}, then 
\[
v(x,m) = \delta^{\theta}\frac{x^{1-\gamma}}{1-\gamma}\tilde{c}_{0}(m)^{-\frac{\theta}{\psi}}\quad \hbox{for $x> 0$.} 
\]
Furthermore, $c_{t}^{*}:= \tilde{c}_{0}(m)X_t$, $\pi^{*}_{t}:= \frac{\mu}{\gamma\sigma^{2}}$, and $h^*_t:= 0$, for $t\ge 0$, form an optimal control for \eqref{problem'}.
\end{proposition}

\begin{proof}
See Section~\ref{subsec:proof of prop:NoAging}. 
\end{proof}

Proposition~\ref{prop:NoAging} shows that without aging and healthcare, optimal investment follows classical Merton's proportion, while the optimal consumption rate is the constant $\tilde{c}_{0}(m)$, dictated by the fixed mortality $m$. By \eqref{c_0}, for the case $\zeta=1$, $\tilde c_0(m)\equiv \psi\delta + (1-\psi)\big(r+\frac{1}{2\gamma}\left(\frac{\mu}{\sigma}\right)^{2}\big) $ no longer depends on $m$. Indeed, with no loss of wealth (and thus utility) at death, dying sooner or later %(i.e. how large $m$ is) 
does not make a difference to one who maximizes lifetime utility plus bequest utility.

As $\frac{\zeta^{1-\gamma}-1}{1-\gamma}<0$ for $\zeta<1$ and $0<\gamma\neq 1$, we observe from \eqref{c_0} that a larger mortality rate $m$ induces a larger consumption rate due to EIS $\psi>1$. %, but a lower consumption rate when $\psi <1$. 
This can be explained by the usual income and substitution effects in response to negative wealth shocks.
A larger mortality rate makes the loss of wealth at death more pressing and imminent. This reduces the total income generated by saving up to the death time, leading to the income effect that reduces consumption in the current period. On the other hand, as saving is now less effective in generating future income, the opportunity cost of consumption in the current period decreases. This brings about the substitution effect that increases current consumption. As is known in the literature, when EIS $\psi>1$, the substitution effect prevails, encouraging the agent to consume more. % (i.e. consumption substitutes for saving). % (substitution effect), mortality shocks also means less bequest left to beneficiaries, which in turn stipulates savings to alleviate the consumption shock (income effect). Each of these countervailing effects prevails above or below $\psi=1$. 

%%%%%%%%%%%%%%%%%%%%%%%%%%%%%

\subsection{Aging without Healthcare}\label{beta>0, g=0}
When the natural growth of mortality is positive ($\beta>0$) but healthcare is unavailable ($g\equiv 0$), mortality grows exponentially, i.e. $M_{t}=me^{\beta t}$. %, consistently with the Gompertz law. 
As $g\equiv 0$ and $v(x,m)$ is nondecreasing in $x$ by definition, the second supremum in \eqref{HJB} vanishes. It follows that \eqref{ODE} reduces to 
\begin{equation}\label{ODE'}
0 = u(m)^{2}-\tilde{c}_{0}(m) u(m)-\beta mu'(m),\quad m>0.   
\end{equation}
This type of differential equations can be solved explicitly.

\begin{lemma}\label{ANH_ODE}
	Fix $\ell>0$, and define the function $u_\ell:\R_+\to\R_+$ by
	\begin{equation}\label{u_q}
	u_{\ell}(m) := \left(\frac{1}{\ell}\int_{0}^{\infty}e^{\frac{\psi-1}{\ell(1-\gamma)}(\zeta^{1-\gamma}-1)my}(y+1)^{-\left(1+\frac{k^*}{\ell}\right)}dy\right)^{-1}.
	\end{equation}
	If $k^*>0$ in \eqref{k^*}, then $u_\ell$ is the unique solution to the ordinary differential equation
	\begin{equation}\label{ANHp_ODE}
	0= u^{2}(m)-\tilde{c}_{0}(m)u(m)-\ell m u'(m),\quad\forall m >0,
	\end{equation}
	such that $\lim_{\ell\rightarrow0}u_{\ell}(m) = \tilde{c}_{0}(m)$. Moreover, $u_\ell$ satisfies
	%\begin{itemize}
		%\item[(a)]
		\[
		u_{\ell}(0) = \tilde{c}_{0}(0) = k^*>0,\quad \lim_{m\rightarrow\infty}\left[u_{\ell}(m) - (\tilde{c}_{0}(m) + \ell)\right] = 0,
		\]
		\begin{equation}\label{ANHbound}
		\tilde{c}_{0}(m)< u_{\ell}(m)<\tilde{c}_{0}(m) + \ell,\quad \forall m>0.
		\end{equation}
		%\item[(b)]$u_{q}'(0+) = \infty$, $u_{q}'(\infty) = (\psi-1)\frac{1-\zeta^{1-\gamma}}{{1-\gamma}}$.
	%\end{itemize}
\end{lemma}

\begin{proof}
%As the results follow from analogous arguments in \cite{Huang19}, we only sketch the proof. First, 
Similarly to (A.8) in \cite{Huang19}, \eqref{ANHp_ODE} admits the general solution
	\[
	u(m) = \ell e^{\frac{\psi}{\theta \ell}(\zeta^{1-\gamma}-1)m}\left(C\beta m^{\frac{k}{\beta}}+ \int_{1}^{\infty}e^{\frac{\psi}{\theta \ell}(\zeta^{1-\gamma}-1)mv}v^{-(1+\frac{k}{\ell})}dv\right)^{-1},\quad \hbox{with}\ C\in\R. 
	\]
	To ensure $\lim_{\ell\rightarrow0}u(m) = \tilde{c}_{0}(m)$, we need $C=0$, which identifies the corresponding solution as
	\begin{align*}
	u_{\ell}(m) &= \ell e^{\frac{\psi}{\theta \ell}(\zeta^{1-\gamma}-1)m}\left( \int_{1}^{\infty}e^{\frac{\psi}{\theta \ell}(\zeta^{1-\gamma}-1)mv}v^{-(1+\frac{k}{\ell})}dv\right)^{-1}.%= \left(\frac{1}{q} \int_{0}^{\infty}e^{\frac{\psi}{\theta q}(\zeta^{1-\gamma}-1)my}(y+1)^{-(1+\frac{k}{q})}dy\right)^{-1}.
	\end{align*}
	A straightforward change of variable then gives the formula \eqref{u_q}. Now, replacing the positive constants $\frac{\delta + (\gamma-1)r}{\gamma}$, $\beta$, and $\frac{1-\zeta^{1-\gamma}}{\gamma}$ in \cite[Lemma A.1]{Huang19} by $k^*$, $\ell$, and $-\frac{\psi-1}{1-\gamma}(\zeta^{1-\gamma}-1)$ in our setting, we immediately obtain the remaining assertions.  
\end{proof}

%Relying on Lemma~\ref{ANH_ODE}, the problem \eqref{problem'} can be solved. 

\begin{proposition}\label{prop:Aging}
Assume $\beta>0$ and $g\equiv 0$. If $k^*>0$ in \eqref{k^*}, then 
\[
v(x,m) = \delta^{\theta}\frac{x^{1-\gamma}}{1-\gamma}u_\beta(m)^{-\frac{\theta}{\psi}},\qquad (x,m)\in \R_+^2, 
\]
where $u_{\beta}:\R_+\to\R_+$ is defined as in \eqref{u_q}, with $\ell=\beta$. Furthermore, $c^*_t:= u_\beta(me^{\beta t})X_t$, $\pi^*_t := \frac{\mu}{\gamma\sigma^{2}}$, and $h^*_t:=0$, for $t\ge 0$, form an optimal control for \eqref{problem'}.
\end{proposition}

\begin{proof}
See Section~\ref{subsec:proof of prop:Aging}. 
\end{proof}

%Proposition~\ref{prop:Aging}, together with Lemma~\ref{ANH_ODE}, admits interesting implications. First, we note 
Observe from \eqref{c_0} and \eqref{k^*} that 
\begin{equation}\label{c_0=k^*+...}
\tilde c_0(m)= k^* + (\psi-1)\frac{(1-\zeta^{1-\gamma})m}{1-\gamma}.
\end{equation}
As $\psi>1$ and $\frac{1-\zeta^{1-\gamma}}{1-\gamma}>0$ for all $0<\gamma\neq 1$,  the condition $k^*>0$ ensures $\tilde c_0(m)>0$ for all $m>0$. This, together with $u_\beta>\tilde c_0$ (\eqref{ANHbound} with $\ell=\beta$), shows that $k^*>0$ in Proposition~\ref{prop:Aging} is essentially a well-posedness condition, ensuring that the optimal consumption rate $u_\beta(me^{\beta t})$ is strictly positive for all $t\ge 0$. Moreover, with $\ell=\beta$, \eqref{ANHbound} stipulates that aging enlarges consumption rate, but the increase does not exceed the growth of aging $\beta>0$; %in addition, this upper bound is asymptotically reached as mortality increases indefinitely. 
note that the increase in consumption results from the same substitution effect as discussed below Proposition~\ref{prop:NoAging}. %Lemma~\ref{ANH_ODE} (b) further describes how the optimal consumption rate increases with mortality: as the agent is young (i.e. $m$ is small), it increases steeply; as the agent is old (i.e. $m$ is large), it grows asymptotically linearly, with the same slope $(\psi-1)\frac{1-\zeta^{1-\gamma}}{{1-\gamma}}$ as in the case without aging. 

%%%%%%%%%%%%%%%%%%%%%%%%%%%%%%%

\subsection{Aging and Healthcare}\label{beta>0, g>0}
For the general case where the natural growth of mortality is positive ($\beta>0$) and healthcare is available ($g\not\equiv 0$), we need to deal with the equation \eqref{ODE} in its full complexity. %To this end, we impose the following condition on $g$. 

\begin{assump}\label{assump:AH}
	Let $g:\R_{+}\to\R_{+}$ be twice differentiable with $g(0)=0$, $g'(h)>0$ and $g''(h)<0$ for $h>0$, and satisfies the Inada condition
	\begin{equation}
	g'(0+) = \infty\quad \text{and}\quad g'(\infty) = 0, 
	\end{equation}
	as well as
	\begin{equation}\label{g<beta}
	g\left(I\left(\psi-1\right)\right)<\beta\quad \text{with}\quad I:= (g')^{-1}.
	\end{equation}
\end{assump}
Condition \eqref{g<beta} was first introduced in \cite{Huang19}. Its purpose will be made clear after the optimal healthcare spending strategy $h^*$ is introduced in Theorem~\ref{Thm:AH}; see Remark~\ref{rem:why g<beta}. 

\begin{lemma}\label{ODE_Solution}
Suppose Assumption \ref{assump:AH} holds. If $k^*>0$ in \eqref{k^*}, there exists a unique nonnegative, strictly increasing, strictly concave, classical solution $u^*:\R_+\to\R_+$ to \eqref{ODE}. 
	Furthermore, define 
	\[
	\underline{\beta}:= \beta - \sup\limits_{h\ge0}\left\{g(h)-(\psi-1)h\right\}\in(0,\beta). %\quad \hbox{and for any $m\ge0$:
	\]
	Then, $\lim_{m\rightarrow\infty}\left[u^*(m) - (\tilde{c}_{0}(m) + \underline\beta)\right] = 0$ and
	\begin{equation}\label{u^* bounds}
	u_{\underline{\beta}}(m)\le u^*(m)\le \min\{u_{\beta}(m),\tilde{c}_{0}(m)+\underline{\beta}\}\quad \forall m>0.
	\end{equation}
\end{lemma}

\begin{proof}
By replacing positive constants $\frac{1-\gamma}{\gamma}$, $\frac{\delta + (1-\gamma)r}{\gamma}$, and $\frac{1-\zeta^{1-\gamma}}{\gamma}$ in \cite[Appendix A.3]{Huang19} (particularly Theorems 3.1 and 3.2) by $\psi-1$, $k^*$, and $ -\frac{\psi-1}{1-\gamma}(\zeta^{1-\gamma}-1)$ in our setting, we get the desired results.
\end{proof}

\begin{remark}
The tractable lower and upper bounds for $u^*$ in \eqref{u^* bounds} will play a crucial role in verification arguments in the proof of Theorem~\ref{Thm:AH} below, as well as calibration in Section~\ref{sec:calibration}. 
\end{remark}

\begin{theorem}\label{Thm:AH}
Suppose Assumption \ref{assump:AH} holds.  If $k^*>0$ in \eqref{k^*}, then 
\begin{equation}\label{v's form}
v(x,m) = \delta^{\theta}\frac{x^{1-\gamma}}{1-\gamma}u^*(m)^{-\frac{\theta}{\psi}},\quad (x,m)\in\R^2_+, 
\end{equation}
where $u^*:\R_{+}\to\R_{+}$ is the unique nonnegative, strictly increasing, strictly concave, classical solution to \eqref{ODE}. Furthermore, $(c^*,\pi^*,h^*)$ defined by
	\[
	c^{*}_{t} := u^*(M_{t})X_t,\quad  \pi_{t}^{*}:= \frac{\mu}{\gamma\sigma^{2}},\quad h^{*}_{t}:= (g')^{-1}\left((\psi-1)\frac{u^{*}(M_{t})}{M_{t}(u^*)'(M_{t})}\right),\qquad t\ge 0
	\]
	is an optimal control for \eqref{problem'}. 
\end{theorem}

\begin{proof}
See Section~\ref{subsec:proof of Thm:AH}. 
\end{proof}

Theorem~\ref{Thm:AH} identifies the marginal efficacy of optimal healthcare spending, $g'(h^*_t)$, to be inversely proportional to $\frac{m(u^*)'(m)}{u^*(m)}$, the elasticity of consumption with respect to mortality, where the constant of proportionality depends on EIS $\psi$. Note that a larger EIS implies less healthcare spending, as $(g')^{-1}$ is strictly decreasing. In a sense, healthcare spending is like saving: it crowds out current consumption, but potentially enlarges future consumption by extending one's lifetime. Since a larger EIS means a stronger substitution effect %of future loss of wealth at death 
(as discussed below Proposition~\ref{prop:NoAging}), one substitutes more consumption for saving-like healthcare spending with a larger $\psi$. 

%Although the optimal consumption rate $u^*(m)$, the solution to \eqref{ODE}, does not admit an explicit formula, it does have simple upper and lower bounds as in \eqref{u^* bounds}, based on constructions in Sections~\ref{beta>0, g=0} and \ref{beta=0, g=0}. %The lower bound is $u_{\underline{\beta}}(m)$, the consumption rate in a model where healthcare is unavailable, but mortality grows at the lower rate $\underline\beta<\beta$, low enough that the agent would be willing to give up access to healthcare in exchange for such a slower natural growth of mortality. The upper bound is the minimum between the consumption rate in the case of aging without healthcare (i.e. $u_\beta(m)$ in Section~\ref{beta>0, g=0}), and the consumption rate in the case of neither aging nor healthcare (i.e. $\tilde c_0(m)$ in Section~\ref{beta=0, g=0}) plus the adjusted growth rate $\underline\beta$.
%These bounds for $u^*$ are crucial for verification arguments in the proof of Theorem~\ref{Thm:AH}, as well as calibration in Section~\ref{sec:calibration}. 

\begin{remark}\label{rem:why g<beta}
As the same argument in \cite[Lemma A.2]{Huang19} implies $\frac{u^*(m)}{ m(u^{*}(m))'}\ge1$ for $m>0$, 
\begin{equation}\label{g(h^*)<beta}
g(h^*_t) = g\left(   I\left((\psi-1)\frac{u^{*}(M_{t})}{M_{t}(u^*)'(M_{t})}\right)  \right)\le g(I(\psi-1)) <\beta, 
\end{equation}
where the last inequality is due to \eqref{g<beta}. In other words, \eqref{g<beta} stipulates that optimizing healthcare spending can only reduce, but not reverse, the growth of mortality.  
\end{remark}

\begin{remark}\label{rem:confirms}
Since the transferred wealth at death is $\zeta X^{c^*,\pi^*,h^*}_{\tau-}$, \eqref{v's form} indicates that %the resulting bequest utility is 
\[
\delta^{\theta}\frac{(\zeta X^{c^*,\pi^*,h^*}_{\tau-})^{1-\gamma}}{1-\gamma}u^*(M^{h^*}_{\tau-})^{-\frac{\theta}{\psi}} = \zeta^{1-\gamma} v(X^{c^*,\pi^*,h^*}_{\tau-}, M^{h^*}_{\tau-}),
\]
i.e. the loss of wealth at death reduces utility by a factor of $\zeta^{1-\gamma}$, confirming the setup in \eqref{Vtilde_Utility}. 
\end{remark}

\begin{remark}
For the case $\psi=1/\gamma>1$, Propositions~\ref{prop:NoAging}, \ref{prop:Aging} and Theorem~\ref{Thm:AH} reduce to results in \cite{Huang19} under time-separable utilities; see Propositions 3.1, 3.2, and Theorems 3.4, 4.1 therein. 
\end{remark}

%%%%%%%%%%%%%%%%%%%%%%%%%%%%%%%%%%%
%%%%%%%%%%%%%%%%%%%%%%%%%%%%%%%%%%%

\section{Calibration: A Preliminary View}\label{sec:calibration}
 In this section, we calibrate the model in Section~\ref{beta>0, g>0} to actual mortality data. We take as given $r=1\%$, $\delta=3\%$, $\psi=1.5$, $\gamma=2$, $\zeta=50\%$, $\mu = 5.2\%$, and $\sigma = 15.4 \%$. 
A safe rate $r=1\%$ approximates the long-term average real rate on Treasury bills in \cite{beeler2012long}, and the time preference $\delta=3\%$ is also consistent with estimates therein; $\psi=1.5$ is estimated in \cite{Bansal04}; $\gamma=2$ follows the specification in \cite{Kraft17} and \cite{Xing}; $\mu = 5.2\%$ and $\sigma = 15.4 \%$ are taken from the long-term study \cite{Imperial18}; $\zeta=50\%$ is a rough estimate of inheritance and estate taxes in developed countries. These values ensure $k^*>0$ in \eqref{k^*}. 
In addition, we take the efficacy function $g:\R_+\to\R_+$ to be 
\begin{equation}\label{g'}
g(z) = a\cdot ({z^{q}}/{q}),\quad \hbox{with}\ a>0\ \hbox{and}\ q\in(0,1).
\end{equation}
The equation \eqref{ODE} then becomes
 \begin{equation}\label{ODE_Example}
 u^{2}(m) - \tilde{c}_{0}(m)u(m) - \beta mu'(m) + ({(1-q)}/{q})a^{\frac{1}{1-q}}\left((\psi-1)u(m)\right)^{\frac{-q}{1-q}}(mu'(m))^{\frac{1}{1-q}}=0, 
 \end{equation}
and the optimal healthcare spending process is now 
$
 h_{t}^* = \big(a^{-1}(\psi-1)\frac{u^*(M_{t})}{M_{t}(u^*)'(M_{t})}\big)^{\frac{-1}{1-q}},
 $
 where $u^*$ is the unique solution to \eqref{ODE_Example}. The endogenous mortality is then
 \begin{equation}\label{Mortality_Example}
 dM_{t} = M_{t}\left(\beta - \frac{1}{q}a^{\frac{1}{1-q}}\left((\psi-1)\frac{u^*(M_{t})}{M_{t}(u^*)'(M_{t})}\right)^{\frac{-q}{1-q}}\right)dt,\quad M_{0} = m_{0}>0.
 \end{equation}
 
We calibrate $\beta>0$, $a>0$, $q\in(0,1)$, and $m_0>0$ to mortality data in the US and UK. %, the Netherlands, and Bulgaria. 
For each country, the natural growth rate of mortality $\beta>0$ is estimated from mortality data for the cohort born in 1900, assuming no healthcare available. Given this estimated $\beta>0$, healthcare parameters $a>0$ and $q\in (0,1)$  in \eqref{g'}, as well as initial mortality $m_0>0$, are calibrated by matching the endogenous mortality curve \eqref{Mortality_Example} with mortality data for the cohort born in 1940, through minimizing the mean squared error (MSE). %Calibration results are listed in Table \ref{table1}.
Essentially, we work under the assumption that the 1900 cohort had no access to healthcare (whence its mortality grew exponentially with the Gompertz law) and the 1940 cohort had full access to healthcare. This is a crude simplification, but conforms to several realistic constraints; see \cite[Section 5.2]{Huang19}. 

It is worth noting that solving \eqref{ODE_Example} directly for $u^*$ is challenging. To the best of our knowledge, the mainstream solvers (e.g. in Mathematica and Matlab) crucially require that the first derivative $u'(m)$ in a first-order ODE be expressed as a function of $u(m)$ and $m$. Such an expression is not available to \eqref{ODE_Example} because of the nonlinearity induced by $g(z)=a z^{q}/q$. In an attempt to circumvent this, we follow \cite[Algorithm 8.1]{Kraft17} to approximate $u^*$ in a recursive manner. The algorithm, however, converges for some specifications of $(a,q)$ and diverges otherwise.\footnote{\cite[Algorithm 8.1]{Kraft17} converges desirably for a typical Epstein-Zin utility maximization problem without the consideration of healthcare. When healthcare is considered, the convergence breaks down due to the added efficacy function $g(z)=a z^{q}/q$.}  This makes it inappropriate for the purpose of calibration, where we need to solve \eqref{ODE_Example}  for a wide range of $(a,q)$ and select the best specification that brings the model-implied mortality closest to data.

In view of this, we settle ourselves with a fairly simple approximate of $u^*$, i.e.
\begin{equation}\label{u^* calibration}
	%u^{*}(m)\approx 
	\overline{u}(m):=\frac{1}{2}\left(u_{\underline{\beta}}(m) + \min\{u_{\beta}(m),\tilde{c}_{0}(m)+\underline{\beta}\}\right),
\end{equation}
which is the average of the upper and lower bounds of $u^*$ in \eqref{u^* bounds}. By Lemmas~\ref{ANH_ODE} %(with $\ell=\underline\beta$) 
and \ref{ODE_Solution}, %we observe from (4.7) that 
\[
%\|u^*(m) -\overline{u}(m)\|_\infty := 
\sup_{m>0}|u^*(m) -\overline{u}(m)|\le \underline{\beta}/2\quad \hbox{and}\quad \lim_{m\downarrow0}|u^*(m)-\overline{u}(m)|=\lim_{m\uparrow\infty}|u^*(m)-\overline{u}(m)| = 0. 
\]
%\[
%\lim_{m\downarrow0}|u^*(m)-\overline{u}(m)|=\lim_{m\uparrow\infty}|u^*(m)-\overline{u}(m)| = 0. 
%\] 
As $\overline{u}(m)$ has an explicit formula for any specification of $(a,q)$, thanks to the formulas \eqref{c_0} and \eqref{u_q}, it facilitates the calibration significantly. 
% This allows us to consider a much larger grid for $(a,q)$ than would be possible using the sequential estimation seen below. 
The results are listed in Table \ref{table1}.

%Specifically, for a fixed $\varepsilon>0$, we hope to find an $\varepsilon$-approximate for $u^*$ through the algorithm:  
%\begin{itemize}[leftmargin=0.2in]
%\item [1.] Take the initial guess $u_{0}(m) : = u_{\underline{\beta}}(m)$ and set $n:=1$.
%\item [2.] Derive $u_n(m)$ from $u_{n-1}(m)$ by solving
%\begin{equation*}
%	u_{n}^{2}(m) - \tilde{c}_{0}(m)u_{n}(m) - \beta mu_{n}'(m) + ({(1-q)}/{q})a^{\frac{1}{1-q}}\left((\psi-1)u_{n}(m)\right)^{\frac{-q}{1-q}}(mu_{n-1}'(m))^{\frac{1}{1-q}}=0,
%\end{equation*}
%with the Neumann boundary conditions $u_{n}'(0)=\infty$ and $u_{n}'(\infty)=0$. 
%\item [3.] If $\|u_{n}-u_{n-1}\|_{\infty}>\varepsilon$, increase $n$ by 1 and repeat step 2 above; otherwise, report $u_n$ as the approximate for $u^*$. 
%\end{itemize} 

\begin{table}[H]
	\begin{center}
		\caption{Calibration Results}
		\begin{threeparttable}
			\begin{tabular}{l|c|c|c|c|c|c} % <-- Alignments: 1st column left, 2nd middle and 3rd right, with vertical lines in between
				\toprule\midrule
				Country & $\beta$ (\%) & $m_{0}\times10^{4}$ & $a$ & $q$ & Model MSE $\times10^{6}$ & MSE $\times10^{6}$ \\
				\hline
				United States (US) & 7.24069 & 1.34995 & 0.19 & 0.61 & 0.0436896 & 0.128984\\
				United Kingdom (UK) & 7.79605 & 0.843827 & 0.19 & 0.60 & 0.0249924 & 0.12755\\
%				Canada (CA) & 7.48756 & 0.891925 & 0.22 & 0.65 & 0.0339878 & 0.0491404\\
				%Netherlands (NL)\tnote{*} & 8.65832 & 0.477551 & 0.16 & 0.53 & 0.0478583 & 0.207779\\ 
				%Bulgaria (BG)\tnote{**} & 8.86593 & 0.892038 & 0.14 & 0.56 & 0.923716 & 2.85819\\
				\bottomrule\addlinespace[-2ex]
			\end{tabular}
			%\begin{tablenotes}\footnotesize
			%	\item[*] Mortality rates impacted during WWII were excluded when calculating $\beta$.
			%	\item[**] Incomplete data for the 1900 cohort. $\beta$ estimated from age range 47-77. 
			%\end{tablenotes}
		\end{threeparttable}
		\label{table1}
	\end{center}
\end{table}
We stress that the calibration performed in this section, based on $\overline u$ in \eqref{u^* calibration}, is only {\it preliminary}. A more sophisticated approximation of $u^*$ is certainly needed for an in-depth, full-fledged calibration. The purpose of our preliminary study is to  demonstrate the potential of our model and possibly draw more attention to this problem for further developments.

\subsection{Results}%{Mortality and Healthcare Spending} 
In Figure \ref{fig:cal_US}, %presents the model performance under two simplifying assumptions. First, we assume that the cohort born in 1900 had no access to healthcare, so that its mortality grew exponentially with the Gompertz law. 
the blue line is obtained by linearly regressing mortality data of the 1900 cohort (blue dots), while 
the red line is the model-implied mortality curve calibrated to mortality data of the 1940 cohort (red dots). 
%Second, we assume that the cohort born in 1940 had full access to healthcare, and calibrate the model-implied mortality (red curve) to actual mortality data (red dots). %Although these assumptions are crude  approximations, they are in place due to several realistic considerations, as explained in \cite[Section 5.2]{Huang19}. 
Clearly, our model reproduces declines in mortality that are very close to ones observed historically. When compared with  \cite[Figure 5.2]{Huang19}, Figure \ref{fig:cal_US} provides a much better fit. This improvement can be attributed to the use of Epstein-Zin utilities (so that $\gamma$ and $\psi$ can {\it both} take empirically relevant values), the inclusion of risky assets, and modifications of calibration methods. %, or a combination of all three. 
Figure~\ref{fig:calibration} shows that our model also performs well for the UK data. %not only the US data, but those from other countries as well.
%The success of our model is not limited to the US data. Figure~\ref{fig:calibration} shows that for the UK, the Netherlands, and Bulgaria, the model-implied mortality (red line) well approximates actual mortality data (red dots). 

We also compare our model performance with linear regression. Indeed, without any idea of healthcare, one can model mortality data of the 1940 cohort by linear regression (as we did for the 1900 cohort). Our model outperforms linear regression: %, across all countries considered. Specifically, 
the sixth column of Table \ref{table1} reports MSEs under our model, significantly smaller than those under linear regression in the seventh column. 

\begin{figure}[h]
\centering
%\begin{subfigure}{.48\textwidth}
%\centering
%\includegraphics[width=.9\linewidth]{{Model_Calibration/Updates/Canada_Mortality.pdf}}
%\caption{Canada}
%\label{fig:cal_CA}
%\end{subfigure}
%\begin{subfigure}{.48\textwidth}
	\centering
\includegraphics[width=.6\linewidth]{{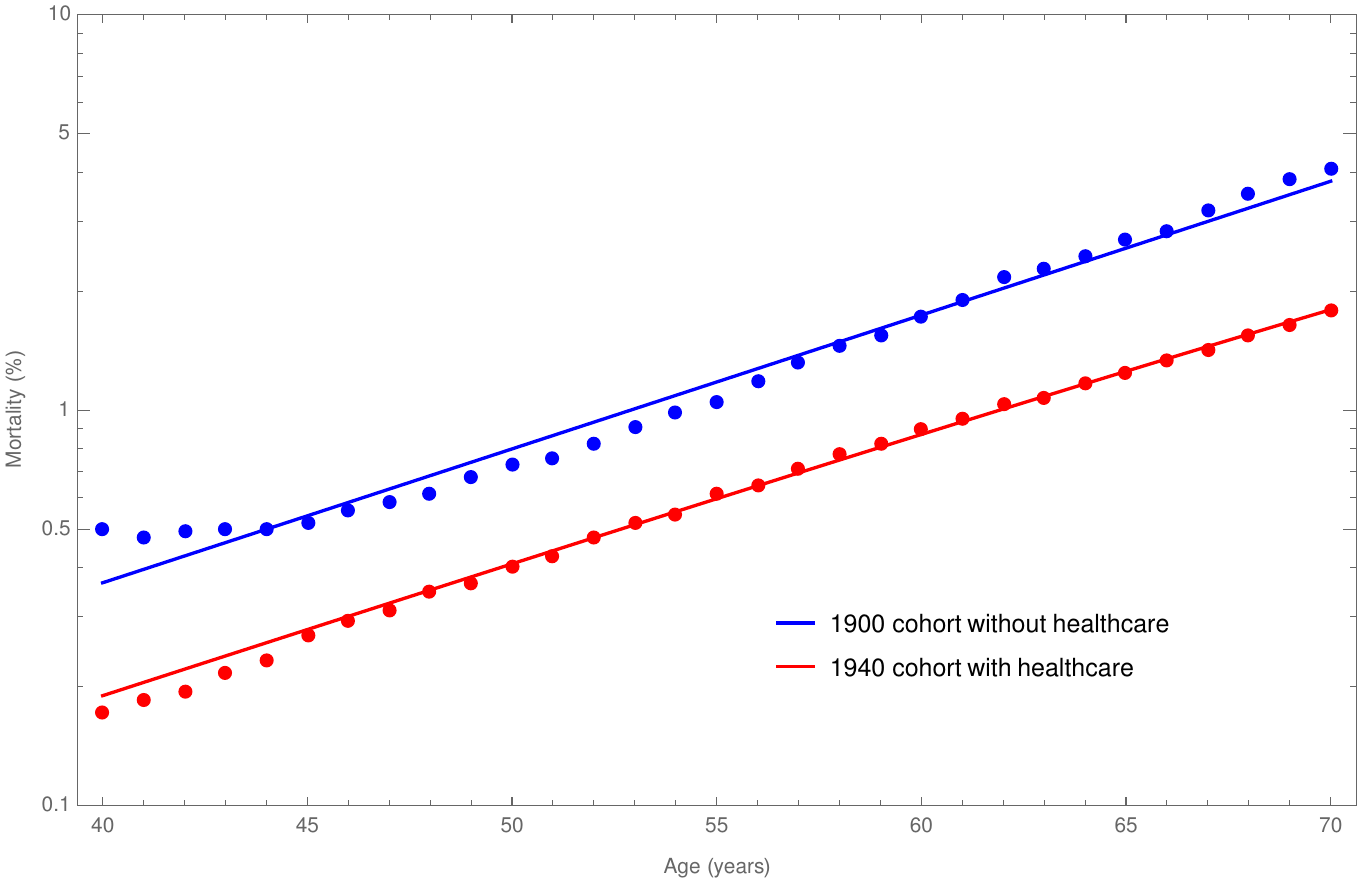}}
%\caption{\small UK}
%\label{fig:cal_UK}
%\end{subfigure}
%\begin{subfigure}{.48\textwidth}
%	\centering
%\includegraphics[width=.9\linewidth]{{Model_Calibration/Updates/Netherlands_Mortality.pdf}}
%\caption{\small Netherlands}
%\label{fig:cal_NE}
%\end{subfigure}\\
%\begin{subfigure}{.48\textwidth}
%	\centering
%\includegraphics[width=.9\linewidth]{{Model_Calibration/Updates/Bulgaria_Mortality.pdf}}
%\caption{\small Bulgaria}
%\label{fig:cal_BG}
%\end{subfigure}%
\caption{\small Mortality rates (log scale) at adults' ages for the cohorts born in 1900 and 1940 in the UK. The dots are actual data (Berkeley Human Mortality Database) and the lines are model-implied mortality curves.}
\label{fig:calibration}
\end{figure}

Figure~\ref{fig: healthcarespending} displays the model-implied optimal healthcare spending.  In both countries, the proportion of wealth spent on healthcare is negligible at age 40, but increases to 0.5-1\% at age 80. Figure \ref{fig:CalibratedEfficiency} presents the calibrated efficacy function $g(h) = a \frac{h^q}{q}$ for the two countries.  It particularly indicates that %indicates a ranking in term of the effectiveness of healthcare spending: 
healthcare is more effective (in reducing mortality growth) in the UK than in the US. 
Along with Figure~\ref{fig: healthcarespending}, we find that lower efficacy of healthcare is compensated by larger healthcare spending relative to wealth. That is, with enhanced efficacy, our model stipulates {\it less} healthcare spending, instead of {\it more} to exploit the reduced marginal cost to curtail mortality growth.

%%%%%%%%%%%%%%%%%%%%%%%%%%

%\subsection{Healthcare Spending}\label{subsec:h^*}
%Figure~\ref{fig: healthcarespending} displays the model-implied optimal healthcare spending in the two countries. The left panel reveals that the proportion of wealth spent on healthcare is negligible at age 40, but increases quickly to 0.5-1\% at age 80. %This is in line with the finding in \cite[Section 5.1]{Huang19} under time-separable utilities. %, and broadly consistent with the US data reported in \cite{HartmanCatlinLassmanEtAl2008}. 
%The right panel further shows that healthcare spending increases with age much faster than consumption and investment combined: it accounts for less than 5\% of total spending at age 40, but increases continuously to 13-30\% at age 80. 
\begin{figure}[h]
	\centering
	\begin{subfigure}{.5\linewidth}
		\centering
		\includegraphics[width=.85\linewidth]{{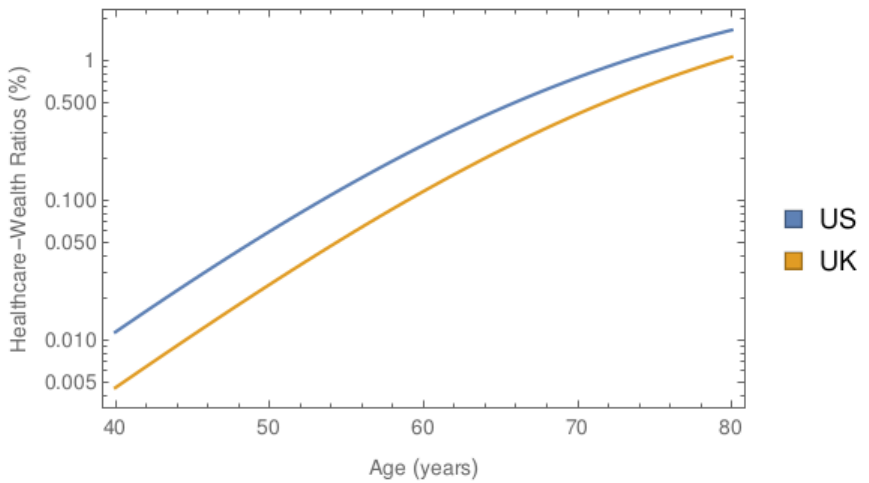}}
		\caption{Healthcare-wealth ratio (log scale) at adult ages.}
		\label{fig: healthcarespending}
	\end{subfigure}
	\begin{subfigure}{.48\linewidth}
		\centering
		\includegraphics[width=.85\linewidth]{{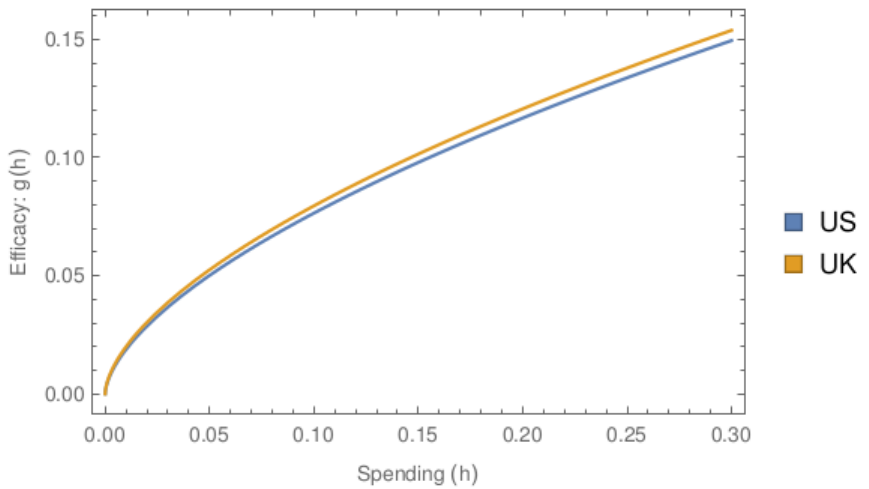}}
		\caption{Efficacy $g(h)$ given healthcare-wealth ratio $h$.}
		\label{fig:CalibratedEfficiency}
	\end{subfigure}
	\caption{\small Calibrated healthcare-wealth ratio and efficacy of healthcare in the US and UK. 
	%Left panel: Optimal healthcare-wealth ratio (vertical, log-scale) at adult ages (horizontal). Right panel: Calibrated efficacy of healthcare $g(h)$, measured by the reduction in the growth of mortality, given proportions of wealth $h$ spent on healthcare.
	}
	%\label{fig: healthcarespending}
\end{figure} 

While Figure \ref{fig:CalibratedEfficiency} hints at the potential of our model as a new analytic tool for healthcare efficacy, we stress that a more in-depth statistical and economic analysis is required here. First, one needs to find the confidence intervals for the estimated $(a,q)$, so as to test the hypothesis that the parameter differences across countries are statistically significant. Second, the economic interpretation of $q$ demands further investigation. While a higher $a$ unambiguously raises efficacy, the effect of $q$ is subtle: The efficacy increases faster with a lower $q$ when $h$ is small, but with a larger $q$ when $h$ is large. % $h$, the efficacy increases faster with a larger $q$.
%the efficacy increases faster  but slower for large $h$, and vice versa with a higher $q$. As a small $h$ more likely covers basic standardized medical treatments while a large $h$ covers advanced sophisticated ones, the effect of $q$ seemingly suggests that the efficacy of basic treatments can be raised more easily than sophisticated ones.  
A careful analysis of these issues is well-warranted but beyond the scope of this paper, and we will leave it for future research.

\appendix
\section{Proofs}\label{sec:proofs} %for Section~\ref{sec:EZ}}%\label{sec:derivation of wtV'}
\subsection{Proof of Proposition~\ref{prop:decompose wtV}}\label{subsec:proof of prop:decompose wtV}
First, we assume that $\widetilde V$ is a $\mathbb G$-adapted semimartingale, with $\bar \E[\sup_{s\in[0,t]}|\widetilde V_{s}|]<\infty$ for all $t\ge 0$, that satisfies \eqref{Vtilde_Utility}. Our goal is to show that $\widetilde V$ must be of the form \eqref{wtV decompose}.
In view of \eqref{tau} and \eqref{exp law}, for any $0\le t \le s$, it holds for $\bar\P$-a.e. $\bar\omega=(\omega,\omega')\in\bar\Omega$ that
\begin{equation}\label{tau>s}
\bar \P(\tau>\ell\mid \F_s\vee\cH_t)(\bar\omega) = e^{-\int_{t}^{\ell}M_{u}^{h}(\omega)du} \ind_{\{\tau>t\}}(\bar\omega),\quad \forall t\le\ell\le s. 
\end{equation}
Also, since $\widetilde V$ is a $\mathbb G$-adapted semimartingale, it follows from \eqref{cG} that there exists an $\mathbb F$-adapted semimartingale $V$ such that
\begin{equation}\label{wtV=V}
\widetilde V_t = V_t\quad \hbox{$\bar\P$-a.s. on $\{t<\tau\}$},\qquad \forall t\ge 0. 
\end{equation}
Indeed, for any fixed $\omega\in\Omega$, consider $A_t(\omega):=\{\omega'\in\Omega': t<\tau(\omega,\omega')\}$ for all $t\ge 0$. As $\widetilde V$ is $\mathbb G$-adapted, \eqref{cG} implies $\widetilde V_t(\omega,\omega')$ is constant $\P'$-a.s. on $A_t(\omega)$. By defining $V_t(\omega) = \widetilde V_t(\omega,\omega')$, with $\omega'\in A_t(\omega)$, for all $t\ge 0$, $V$ is an $\mathbb F$-adapted semimartingale satisfying \eqref{wtV=V}. Also note that $\E[\sup_{s\in[0,t]}|V_{s}|]<\infty$, as $\bar \E[\sup_{s\in[0,t]}|\widetilde V_{s}|]<\infty$, for all $t\ge 0$. 
Now, observe that
\begin{align}
\bar{\E}\left[\int_{t\wedge\tau}^{T\wedge\tau}f(c_s,\widetilde{V}_{s}^{c,h})ds\ \middle|\ \cG_t\right]&=\bar \E\left[\int_{t}^{T}\ind_{\{s<\tau\}}f(c_s,\widetilde{V}_{s}^{c,h})ds\ \middle|\ \F_t\vee\cH_t\right]\nonumber\\
&=  \int_{t}^{T}\bar \E\left[\ind_{\{s<\tau\}}f(c_s,\widetilde{V}_{s}^{c,h})\ \middle|\ \F_t\vee\cH_t\right]ds\nonumber\\
&=  \int_{t}^{T}\bar \E\left[\bar \E\big[\ind_{\{s<\tau\}}f(c_s,{V}_{s}^{c,h})\mid \F_s\vee\cH_t\big]\ \middle|\ \F_t\vee\cH_t\right]ds\nonumber\\
&=  \int_{t}^{T}\bar \E\left[f(c_s,{V}_{s}^{c,h})\ \bar \E\big[\ind_{\{s<\tau\}}\mid \F_s\vee\cH_t\big]\ \middle|\ \F_t\vee\cH_t\right]ds\nonumber\\
&= \int_{t}^{T}\bar \E\left[f(c_s,{V}_{s}^{c,h}) \ind_{\{t<\tau\}}e^{-\int_{t}^{s}M_{u}^{h}du}\ \middle|\ \F_t\vee\cH_t\right]ds\nonumber\\
&= \bar \E\left[\int_{t}^{T}\ind_{\{t<\tau\}}e^{-\int_{t}^{s}M_{u}^{h}du}f(c_s,{V}_{s}^{c,h})ds\ \middle|\ \cG_t\right],\label{11}
\end{align}
where the second and last equalities follow from Fubini's theorem for conditional expectations (see \cite[Theorem 27.17]{Schilling-book-2017}), the third equality is due to the tower property of conditional expectations and \eqref{wtV=V}, the fourth equality results from $c_s\in \F_s$ and $V^{c,h}_s\in \F_s$, and the fifth equality holds thanks to \eqref{tau>s}. Next, for $\bar\P$-a.e. fixed $\bar\omega=(\omega,\omega')\in\bar\Omega$, consider the cumulative distribution function of $\tau$ given the information $\F_T\vee\cH_t$, i.e.
\[
F(s):= \bar\P(\tau\le s\mid \F_T\vee\cH_t)(\bar\omega),\quad s\ge0. 
\]
Thanks to \eqref{tau>s}, $F(s) = 1- e^{-\int_{t}^{s}M_{u}^{h}(\omega)du} \ind_{\{\tau>t\}}(\bar\omega)$ for $t\le s\le T$. This implies
\begin{equation}\label{tau density}
\eta(s) := F'(s) = M_{s}^{h}(\omega)e^{-\int_{t}^{s}M_{u}^{h}(\omega)du} \ind_{\{\tau>t\}}(\bar\omega),\quad \hbox{for}\ t\le s\le T, 
\end{equation}
which is the density function of $\tau$ given the information $\F_T\vee\cH_t$. It follows that
\begin{align}
\bar\E\left[\widetilde{V}_{\tau-}^{c,h}\ind_{\{\tau\le T\}}\ \middle|\ \cG_t\right] &= \bar\E\left[{V}_{\tau-}^{c,h}\ind_{\{\tau\le T\}}\ \middle|\ \cG_t\right] \ind_{\{\tau\le t\}} +\bar\E\left[{V}_{\tau-}^{c,h}\ind_{\{\tau\le T\}}\ \middle|\ \cG_t\right] \ind_{\{\tau>t\}}\nonumber\\
&={V}_{\tau-}^{c,h}\ind_{\{\tau\le t\}}+ \bar\E\left[ \bar\E\big[{V}_{\tau-}^{c,h}\ind_{\{t<\tau\le T\}}\mid \F_T\vee\cH_t\big]\ \middle|\ \F_t\vee\cH_t\right]\nonumber\\
&={V}_{\tau-}^{c,h}\ind_{\{\tau\le t\}}+ \bar\E\left[ \int_{t}^{T}\ind_{\{t<\tau\}}M_{s}^{h}e^{-\int_{t}^{s}M_{u}^{h}du}{V}_{s}^{c,h}ds\ \middle|\ \cG_t \right],\label{22}
\end{align}
where the first line results from $\widetilde V_{\tau-}=V_{\tau-}$ (by \eqref{wtV=V}), the second line follows from the tower property of conditional expectations, and the third line is due to the density formula \eqref{tau density}. Since $V$ is right-continuous, it has at most countably many jumps on $[t,T]$, so that we may use $V_s$ (instead of $V_{s-}$) in the last term of \eqref{22}. Finally, 
\begin{align}
&\bar\E\left[\widetilde{V}_{T}^{c,h}\ind_{\{\tau> T\}}\ \middle|\ \cG_t \right] = \bar\E\left[ \bar\E\big[{V}_{T}^{c,h}\ind_{\{\tau>T\}}\mid \F_T\vee\cH_t\big]\ \middle|\ \F_t\vee\cH_t\right]\nonumber\\
 &\hspace{0.7in}=\bar\E\left[ {V}_{T}^{c,h} \bar\E\big[\ind_{\{\tau>T\}}\mid \F_T\vee\cH_t\big]\ \middle|\ \F_t\vee\cH_t\right]%\nonumber\\
 = \bar \E\left[\ind_{\{t<\tau\}}e^{-\int_{t}^{T}M_{u}^{h}du}{V}_{T}^{c,h}\ \middle|\ \cG_t \right],\label{33}
\end{align}
where the first equality follows from the tower property of conditional expectations and \eqref{wtV=V}, the second equality is due to ${V}_{T}\in \F_T$, and the third equality is a consequence of \eqref{tau>s}. Now, combining \eqref{11}, \eqref{22}, and \eqref{33}, we obtain from \eqref{Vtilde_Utility} and $\widetilde V_{\tau-}=V_{\tau-}$ that
\begin{align}
\widetilde{V}_{t}^{c,h}
& =\E_{t}\bigg[\int_{t}^{T}e^{-\int_{t}^{s}M_{r}^{h}dr}\left(f(c_{s},{V}_{s}^{c,h}) + \zeta^{1-\gamma} M_{s}^{h}{V}_{s}^{c,h}\right)ds + e^{-\int_{t}^{T}M_{s}^{h}ds}{V}_{T}^{c,h} \bigg]  \ind_{\{t<\tau\}}\nonumber\\
&\hspace{2.8in}+ \zeta^{1-\gamma}{V}_{\tau-}^{c,h}\ind_{\{t\ge\tau\}},\ \ \text{ for all }0\le t\le T<\infty,\label{wtV'}
\end{align}
where we use the notation $\E_{t}\left[\cdot\right]=\E\left[\cdot|\F_{t}\right]$. 
This, together with \eqref{wtV=V}, particularly implies
\begin{equation}\label{with ind}
V_t(\omega) \ind_{\{t<\tau\}(\omega,\omega')} = \widetilde V_t(\omega,\omega') \ind_{\{t<\tau\}(\omega,\omega')} = E_{t,T}(\omega)\ind_{\{t<\tau\}(\omega,\omega')},
\end{equation}
where
\[
E_{t,T}(\omega) := \E_{t}\bigg[\int_{t}^{T}e^{-\int_{t}^{s}M_{r}^{h}dr}\left(f(c_{s},{V}_{s}^{c,h}) + \zeta^{1-\gamma}M_{s}^{h}{V}_{s}^{c,h}\right)ds + e^{-\int_{t}^{T}M_{s}^{h}ds}{V}_{T}^{c,h}\bigg](\omega). 
\]
For any $\omega\in\Omega$, since there exists $\omega'\in\Omega'$ such that $\ind_{\{t<\tau\}(\omega,\omega')}=1$ (in view of \eqref{tau} and \eqref{exp law}), we conclude from \eqref{with ind} that $V_t(\omega) = E_{t,T}(\omega)$. We can then simplify \eqref{wtV'} as
\begin{align}\label{wtV'''}
\widetilde{V}_{t}
& =V_t \ind_{\{t<\tau\}} + \zeta^{1-\gamma}{V}_{\tau-}\ind_{\{t\ge\tau\}},%\ \ \text{ for all }0\le t\le T<\infty.
\end{align}
where $V$ satisfies
\begin{equation}\label{V eqn}
V_t = \E_{t}\bigg[\int_{t}^{T}e^{-\int_{t}^{s}M_{r}^{h}dr}\left(f(c_{s},{V}_{s}) + \zeta^{1-\gamma} M_{s}^{h}{V}_{s}\right)ds + e^{-\int_{t}^{T}M_{s}^{h}ds}{V}_{T}\bigg],\ \ \forall 0\le t\le T<\infty
\end{equation}
Now, note that the above equation directly implies
\[
V'_t := e^{-\int_{0}^{t}M_{r}^{h}dr}V_t = \mathscr M'_t - \int_0^t e^{-\int_{0}^{s}M_{r}^{h}dr} \left(f(c_{s},{V}_{s}) + \zeta^{1-\gamma} M_{s}^{h}{V}_{s}\right)ds,
\]
where 
\[
\mathscr M'_t := \E_{t}\bigg[\int_{0}^{T}e^{-\int_{0}^{s}M_{r}^{h}dr}\left(f(c_{s},{V}_{s}^{}) + \zeta^{1-\gamma} M_{s}^{h}{V}_{s}^{}\right)ds + e^{-\int_{0}^{T}M_{s}^{h}ds}{V}_{T}^{}\bigg]
\]
is an $\mathbb F$-martingale on $[0,T]$, thanks to %$\E[\sup_{t\in[0,T]}|V_t|]<\infty$ and 
\eqref{V eqn}. 
Applying generalized It\^{o}'s formula for semimartingales (see \cite[Theorem I.4.57]{Jacod03}) to $V_t = e^{\int_{0}^{t}M_{r}^{h}dr} V'_t$ gives
$
dV_t = -F(c_t,M^h_t,V_t) + e^{\int_{0}^{t}M_{r}^{h}dr} d\mathscr M'_t. 
$
Since $0\le M^h_t \le me^{\beta t}$ by definition (by \eqref{Mortality}), $\mathscr M_t:= \int_0^t e^{\int_{0}^{s}M_{r}^{h}dr} d\mathscr M'_s$ is again an $\mathbb F$-martingale. Hence, $V$ is a solution to BSDE \eqref{BSDE_V}. This, together with \eqref{wtV'''}, yields the desired result. 

Next, we prove the converse, i.e. a process $\widetilde V$ given by \eqref{wtV decompose} has the three properties: (i) it is a $\mathbb G$-adapted semimartingale; (ii) $\bar \E[\sup_{s\in[0,t]}|\widetilde V_{s}|]<\infty$ for all $t\ge 0$; (iii) it satisfies \eqref{Vtilde_Utility}. By the construction in \eqref{wtV decompose}, properties (i) and (ii) follow directly from $V$ being an $\mathbb F$-adapted semimartingale with $\E[\sup_{s\in[0,t]}|V_{s}|]<\infty$ for all $t\ge 0$. Now, by applying generalized It\^{o}'s formula for semimartingales (see \cite[Theorem I.4.57]{Jacod03}) to $e^{-\int_{0}^{t}M_{s}^{h}ds}V_{t}$, we see that $V$ satisfies \eqref{V eqn}. This, together with the same arguments in \eqref{11}, \eqref{22}, and \eqref{33}, shows that $\widetilde{V}$ in \eqref{wtV decompose} satisfies \eqref{Vtilde_Utility}.
%noticing the transformed martingale term remains an $\mathbb{F}$-martingale under the boundedness of $e^{-\int_{0}^{s}M_{r}^{h}dr}$. 

%%%%%%%%%%%%%%%%%%%%%%%%%%%%%%%%%

\subsection{Derivation of Proposition~\ref{prop:comparison}}\label{subsec:proof of comparison}

%Recall the generator $F$ defined in \eqref{F func}.

\begin{lemma}\label{Lemma1}
	Let $c, h, V$ and $W$ be $\mathbb F$-progressively measurable processes with $W_{s}\le V_{s}$ for all $s\ge 0$. %Suppose $V_{s},W_{s}$ satisfy \eqref{TransversalityCondition} under $M_{s}$ and $W_{s}\le V_{s}$. 
	If there exists $k\in\R$ such that $V$ satisfies \eqref{comparison_condition}, then
	\begin{equation}\label{monotonicity_V}
	F(c_{s},M^h_{s},V_{s}) - F(c_{s},M^h_{s},W_{s})\le -\Gamma(\Lambda, M^h_{s})(V_{s}-W_{s}),
	\end{equation}
	where $F$ is given in \eqref{F func}, $\Lambda:=\delta\theta +(1-\theta)k$ (as in Definition~\ref{cE}), and $\Gamma$ is defined by
	\begin{equation}\label{Gamma}
	\Gamma(\lambda, m):= \lambda+\frac{\gamma(\psi-1)}{1-\gamma}(1-\zeta^{1-\gamma})m.
	\end{equation}  
\end{lemma}

\begin{proof}
	As in the proof of \cite[Lemma B.1]{Melnyk17}, \eqref{monotonicity_V} holds by the mean value theorem provided that $F_{v}(c_{s},M^h_{s},u)\le -\Gamma(\Lambda, M^h_{s})$ for all $u\in[W_{s},V_{s}]$. To this end, note that
	\[
	F_{v}(c_{s},M^h_{s},u) = -\bigg(\delta\theta + (1-\zeta^{1-\gamma})M^h_{s} + \delta(1-\theta)\bigg(\frac{c_{s}^{1-\gamma}}{(1-\gamma)u}\bigg)^{{1}/{\theta}}\bigg).
	\]
Thanks to Assumption~\ref{specification}, a direct calculation shows $F_{vv}(c_{s},M^h_{s},u)>0$, i.e. $F_{v}(c_{s},M^h_{s},u)$ is increasing in $u$. This, together with $V$ satisfying \eqref{comparison_condition}, implies that for all $u\in[W_{s},V_{s}]$, $F_{v}(c_{s},M^h_{s},u)\le F_{v}(c_{s},M^h_{s},\hat u)$, where $\hat u := \delta^{\theta}\big(k-\frac{\psi-1}{1-\gamma}(\zeta^{1-\gamma}-1)M^h_{s}\big)^{-\theta}\frac{c_{s}^{1-\gamma}}{1-\gamma}$. By direct calculation,
		\begin{align*}
F_{v}(c_{s},M^h_{s},\hat u) &= -\bigg(\delta\theta + (1-\zeta^{1-\gamma})M^h_{s} + (1-\theta)\left(k-\frac{\psi-1}{1-\gamma}(\zeta^{1-\gamma}-1)M^h_{s}\right)\bigg)\\
&= -\left(\Lambda+\frac{\gamma(\psi-1)}{1-\gamma}(1-\zeta^{1-\gamma})M^h_{s}\right)=-\Gamma(\Lambda, M^h_{s}),
	\end{align*}
	where the second equality follows from the definition of $\Lambda$ and $\theta=\frac{1-\gamma}{1-1/\psi}$. 
\end{proof}

To prove Proposition~\ref{prop:comparison}, we intend to follow the idea in the proof of \cite[Theorem 2.2]{Melnyk17}. The involvement of the controlled mortality $M^h$ in \eqref{TransversalityCondition}, as well as the possibility that $\Lambda$ therein can be negative (Remark~\ref{rem:negative Lambda}), result in additional technicalities. The proof below combines arguments in \cite[Theorem 2.2]{Melnyk17} and \cite[Theorem 2.1]{Fan15}, adapted to weaker regularity of processes.

\begin{proof}[Proof of Proposition~\ref{prop:comparison}] 
Recall the function $\Gamma$ in \eqref{Gamma}. Fix $0\le t_0<T$, define 
\begin{equation}\label{Delta}
\Delta_{t}:=e^{-\int_{t_0}^{t}\Gamma(0,M^h_{s})ds}\left(V_{t}^{1}-V_{t}^{2}\right),\quad t\in[t_{0},T], 
\end{equation}
and consider the stopping time $\theta := \inf\left\{s\ge t_0: V_{s}^{1}\le V_{s}^{2}\right\}$. Applying generalized It\^{o}'s formula (see \cite[Theorem I.4.57]{Jacod03}) to $e^{-\int_{0}^{t}\Gamma(0,M^h_{s})ds}V_{t}^{i}$, $i=1,2$, yields
%$h(t,v):=e^{-\int_{0}^{t}\Gamma(0,M_{s})ds}v$ under $V^{1,2}$ to find:
	%\begin{align}
	%d\left(h(t,V_{t}^{1,2})\right) &= h_{t}(t,V_{t}^{1,2})dt + h_{v}(t,V_{t^{-}}^{c,h})dV_{t}^{1,2} + \frac{1}{2}h_{vv}(t,V_{t^{-}}^{1,2})d\left[\overline{V}_{t}^{1,2},\overline{V}_{t}^{1,2}\right]\\
	%&+ \left[h(t,V_{t}^{1,2})-h(t,V_{t^{-}}^{c,h}) - h_{v}(t,V_{t^{-}}^{c,h})(V_{t}^{1,2}-V_{t^{-}}^{1,2})\right]
	%\end{align}
	%where $\overline{V}^{1,2}$ is the continuous part of $V^{1,2}$. Since, $h_{vv}(t,v) = 0$ and $h_{v}(t,v) = e^{-\int_{0}^{t}M_{s}ds}$ the above reduces to:
	\begin{align*}
	d\left(e^{-\int_{0}^{t}\Gamma(0,M^h_{s})ds}V_{t}^{1}\right) &= -e^{-\int_{0}^{t}\Gamma(0,M^h_{s})ds}\left[\Gamma(0,M^h_{s})V_{t}^{1} +F(c_{t},M^h_{t},V_{t}^{1})\right]dt + e^{-\int_{0}^{t}\Gamma(0,M^h_{s})ds}d\mathscr{M}_{t}^{1},\\
	d\left(e^{-\int_{0}^{t}\Gamma(0,M^h_{s})ds}V_{t}^{2}\right) &= -e^{-\int_{0}^{t}\Gamma(0,M^h_{s})ds}\left[\Gamma(0,M^h_{s})V_{t}^{2} +G(t,V_{t}^{2})\right]dt + e^{-\int_{0}^{t}\Gamma(0,M^h_{s})ds}d\mathscr{M}_{t}^{2},
	\end{align*}
	where $\mathscr M^1$, $\mathscr M^2$ are some $\mathbb F$-martingales on $[0,T]$. As $0\le \Gamma(0,M^h_{t})\le \frac{\gamma(\psi-1)}{1-\gamma}(1-\zeta^{1-\gamma})me^{\beta t}$ by the definition of $M^h$ in \eqref{Mortality},  $r\mapsto \int_{t_{0}}^{r}e^{-\int_{0}^{t}\Gamma(0,M^h_{s})ds}d\mathscr{M}_{t}^{i}$ is a true martingale for $i=1,2$. Hence,
	\begin{align*}
	\Delta_{t}= \E_{t}\left[\int_{t}^{T}\ind_{\{s<\theta\}}\left[\left(F(c_{s},M^h_{s},V_{s}^{1})-G(s,V_{s}^{2})\right) + \Gamma(0,M^h_{s})\left(V_{s}^{1}-V_{s}^{2}\right)\right]e^{-\int_{t_0}^{s}\Gamma(0,M^h_{r})dr}ds + \Delta_{T\wedge\theta}\right].
	\end{align*}
	Observe that
	\begin{align*}
	\ind_{\{s<\theta\}}\left(F(c_{s},M^h_{s},V_{s}^{1})-G(s,V_{s}^{2})\right)&= \ind_{\{s<\theta\}}\left(F(c_{s},M^h_{s},V_{s}^{1})-F(c_{s},M^h_{s},V_{s}^{2})\right)\\ 
	& \hspace{0.5in}+\ind_{\{s<\theta\}}\left(F(c_{s},M^h_{s},V_{s}^{2}) - G(s,V_{s}^{2})\right)\\
	&\le \ind_{\{s<\theta\}}\left(F(c_{s},M^h_{s},V_{s}^{1})-F(c_{s},M^h_{s},V_{s}^{2})\right)\\
	&\le \ind_{\{s<\theta\}}\left( -\Gamma(\Lambda,M^h_{s})\left(V_{s}^{1}-V_{s}^{2}\right)\right),
	\end{align*}
	where the first inequality follows from $F(c_{s},M^h_{s},V_{s}^2)\le G(s,V_{s}^2)$, and the second is due to Lemma \ref{Lemma1}, which is applicable here as $V_{s}^{1}>V_{s}^{2}$ for $s\in[t,\theta)$. Thanks to the above inequality, 
	\begin{align}\label{B1}
	\Delta_{t}&\le \E_{t}\left[\int_{t}^{T}\ind_{\{s<\theta\}}\left[-\Gamma(\Lambda,M^h_{s}) + \Gamma(0,M^h_{s})\right]\left(V_{s}^{1}-V_{s}^{2}\right)e^{-\int_{t_{0}}^{s}\Gamma(0,M^h_{r})dr}ds + \Delta_{T\wedge\theta}\right]\nonumber\\
	& = \E_{t}\left[-\int_{t}^{T} \ind_{\{s<\theta\}}\Lambda\Delta_{s}ds + \Delta_{T\wedge\theta}\right],
	\end{align}
	where the second line follows from $\Gamma(\Lambda,M^h_{s}) = \Lambda + \Gamma(0,M^h_{s})$ and \eqref{Delta}. Multiplying both sides by $\ind_{\{t<\theta\}}$ yields
	\[
	\Delta_{t}\ind_{\{t<\theta\}}\le \E_{t}\left[-\int_{t}^{T}\Lambda\Delta_{s}\ind_{\{s<\theta\}}ds + \Delta_{T\wedge \theta}\ind_{\{t<\theta\}}\right]\le \E_{t}\left[-\int_{t}^{T}\Lambda\Delta_{s}\ind_{\{s<\theta\}}ds + \Delta_{T}\ind_{\{T<\theta\}}\right],
	\]
where the second inequality follows from the right continuity of $V^1$ and $V^2$. Indeed, the right continuity implies $V_{\theta}^{1}\le V_{\theta}^{2}$, so that
	%\begin{equation}\label{B2}
	$\Delta_{T\wedge\theta} = \Delta_{\theta}\ind_{\{\theta\le T\}} + \Delta_{T} \ind_{\{T<\theta\}} \le \Delta_{T} \ind_{\{T<\theta\}}$.
	%\end{equation}
	Set $\Delta_{t}^{+} := \Delta_{t}\ind_{\{t<\theta\}}$, and write the previous inequality as
	%\begin{equation}\label{before GW}
	$\Delta_{t}^{+}\le \E_{t}\big[-\int_{t}^{T}\Lambda\Delta_{s}^{+}ds + \Delta_{T}^{+}\big].$
	%\end{equation}
	Taking expectations on both sides and using Fubini's theorem give 
	\begin{equation}\label{GW}
	\Theta_{t} \le -\int_{t}^{T}\Lambda\Theta_{s}ds + \Theta_{T},
	\end{equation} 
	where
	%Now, consider 
	$\Theta_{t}:= \E\left[\Delta^{+}_{t}\right]\ge 0$ is well-defined as $\Gamma(0,M_{s})\ge0$ and $\E\big[\sup_{t\in[0,T]}|V_{t}^{i}|\big]<\infty$, thanks to $V^i\in\cE^h_k$ (Definition~\ref{cE}), for $i=1,2$. %We then obtain from \eqref{before GW} that 
	Now, if $\Lambda>0$, by writing $\Theta_{T}\ge \Theta_{t} + \int_{t}^{T}\Lambda\Theta_{s}ds$, we apply standard Gronwall's inequality to get $\Theta_{T}\ge \Theta_{t}e^{\int_{t}^{T}\Lambda ds}$, or equivalently 
  \begin{equation}\label{Theta Bound}
	\Theta_{t}\le \Theta_{T}e^{-\int_{t}^{T}\Lambda ds},\quad t\in[t_0,T].
	\end{equation}
If $\Lambda<0$, applying backward Gronwall's inequality (see \cite[Proposition 2]{Fan18}) to \eqref{GW} also gives \eqref{Theta Bound}. By \eqref{Theta Bound}, \eqref{Delta}, and \eqref{Gamma}, we obtain
	\begin{equation}
	\Theta_{t_{0}}\le \Theta_{T}e^{-\int_{t_{0}}^{T}\Lambda ds}\le \E\left[e^{-\int_{t_{0}}^{T}\Gamma(\Lambda,M_{s})ds}\left(|V_{T}^{1}|+|V_{T}^{2}|\right)\right].
	\end{equation}
	Since $T>0$ is arbitrary, the transversality condition in $\eqref{TransversalityCondition}$ for $V_{t}^{1}$ and $V_{t}^{2}$ immediately implies
	\begin{equation}
	0\le \Theta_{t_{0}} \le \lim\limits_{T\rightarrow\infty}\E\left[e^{-\int_{t_{0}}^{T}\Gamma(\Lambda,M_{s})ds}\left(|V_{T}^{1}|+|V_{T}^{2}|\right)\right]= 0.
	\end{equation}
	That is, $\Theta_{t_0}=\E\left[\left(V_{t_0}^{1}-V_{t_0}^{2}\right)1_{\{t_0<\theta\}}\right] = 0$. This entails $\theta=t_0$, and thus $V_{t_{0}}^{1}\le V_{t_{0}}^{2}$. Since $t_0\ge 0$ is arbitrary, we conclude that $V_{t}^{1}\le V_{t}^{2}$ for all $t\ge 0$. 
	\end{proof}

%%%%%%%%%%%%%%%%%%%%

%%%%%%%%%%%%%%%%%%%%%%%%%%%%%%
%%%%%%%%%%%%%%%%%%%%%%%%%%%%%%

%\section{Proofs for Section \ref{sec:Results}}

%\subsection{Proof of Lemma \ref{ANH_ODE}}\label{subsec:proof of Lemma ANH_ODE}
%To prove properties (a) and (b), we follow Lemma A.1 in \cite{Huang19}. Since $u_{q}(m)$ solves \eqref{ANHp_ODE}, notice that for any $m>0$:
	%\begin{equation}\label{ANHp_diff}
	%u_{q}(m)-\tilde{c}_{0}(m) = \frac{q m u_{q}'(m)}{u_{q}(m)}\ge0.
	%\end{equation}	
	%Since $u_{q}(m)$ is positive and strictly increasing, we have $u_{q}(m)\ge\tilde{c}_{0}(m)$. Alternatively, using the bound $(y+1)<e^{y}$ for $y>0$ gives
	%\[
	%u_{q}(m) = \left[\frac{1}{q}\int_{0}^{\infty}e^{\frac{\psi}{\theta q}(\zeta^{1-\gamma}-1)my}(y+1)^{-\left(1+\frac{k}{q}\right)}dy\right]^{-1}<\left[\frac{1}{q}\int_{0}^{\infty}e^{\frac{\psi}{\theta q}(\zeta^{1-\gamma}-1)my-\left(1+\frac{k}{q}\right)y}dy\right]^{-1}.
	%\]
	%However, $\left[\frac{1}{q}\int_{0}^{\infty}e^{\frac{\psi}{\theta q}(\zeta^{1-\gamma}-1)my-\left(1+\frac{k}{q}\right)y}dy\right]^{-1} = \tilde{c}_{0}(m)+q$ yielding $u_{q}(m)<\tilde{c}_{0}(m)+q$. These together give
	%\begin{equation*}
	%\tilde{c}_{0}(m)<u_{q}(m)<\tilde{c}_{0}(m)+q,\text{ for all }m>0.
	%\end{equation*}
	%The remaining statements are identical to the arguments of Lemma A.1 in \cite{Huang19} with the constant $\frac{\delta + (\gamma-1)r}{\gamma}>0$ replaced by $k>0$, $\beta>0$ with $q>0$, and $\frac{1-\zeta^{1-\gamma}}{\gamma}>0$ by $-\frac{\psi}{\theta}(\zeta^{1-\gamma}-1)>0$.

\subsection{Proof of Proposition \ref{prop:NoAging}}\label{subsec:proof of prop:NoAging}
For any fixed $m>0$ such that $\tilde c_0(m)>0$, define $w(x):=\delta^{\theta}\frac{x^{1-\gamma}}{1-\gamma}\tilde{c}_{0}(m)^{-\frac{\theta}{\psi}}$ for $x>0$. In order to apply Theorem~\ref{Verification}, we need to verify all its conditions. It can be checked directly that $w$, as a one-variable function, solves \eqref{HJB} in a trivial way, with all derivatives in $m$ being zero.  
For any $(c,\pi,h)\in\cP=\cP_1$, since $(X^{c,\pi,h})^{1-\gamma}$ satisfies \eqref{TransversalityCondition} (with $\Lambda^*$ in place of $\Lambda$), so does
$
w(X^{c,\pi,h}_t)%=\delta^{\theta}\frac{(X^{c,\pi,h}_t)^{1-\gamma}}{1-\gamma}\tilde{c}_{0}(m)^{-\frac{\theta}{\psi}}
$, i.e. $w(X^{c,\pi,h}_t)\in\cE^h_{k^*}$. By the definitions of $\cP$ and $w$,  $\cP=\cP_1\subseteq \cH_{k^*}$ and \eqref{veri condition} is satisfied.  
As $\tilde c_0(m)>0$, $w_x>0$ and $w_{xx}<0$ by definition. It follows that $\bar c(x,m):=x\tilde{c}_{0}(m)$ and $\bar \pi(x,m) :=\frac{\mu}{\gamma\sigma^2}$ are unique maximizers of the supremums in \eqref{sup1, 2}, respectively. The supremum in \eqref{sup3} is zero, as $g\equiv 0$ and $w_x>0$. Hence, $\bar h(x,m) := 0$ trivially maximizes \eqref{sup3}. The only condition that remains to be checked is ``$(c^*,\pi^*,h^*)$ in \eqref{optimal strategies} belongs to $\cP$ and $W^*_t:= w(X^{c^*,\pi^*,h^*}_t)$ satisfies \eqref{comparison_condition}''.  

Observe that a unique solution $X^*=X^{c^*,\pi^*,h^*}$ to \eqref{Wealth} exists as a geometric Brownian motion %satisfying the dynamics
\begin{equation}\label{X^* GBM}
dX_{t}^{*} = X_{t}^{*}\left(r+\frac{1}{\gamma}\left(\frac{\mu}{\sigma}\right)^{2} - \tilde{c}_{0}(m)\right)dt + X_{t}^{*}\frac{\mu}{\gamma\sigma}dB_{t},
\end{equation}
%	\begin{equation*}
%	X^{*}_t= x\cdot \text{exp}\left(\left(r+\frac{1}{2\gamma}\left(\frac{\mu}{\sigma}\right)^{2} - \tilde{c}_{0}(m) - \frac{(1-\gamma)}{2\gamma^{2}}\left(\frac{\mu}{\sigma}\right)^{2}\right)t + \frac{\mu}{\gamma\sigma}B_{t}\right).
%	\end{equation*}
This implies that
	\begin{align}\label{X^*}
	(X^{*}_t)^{1-\gamma}= x^{1-\gamma}\text{exp}\left((1-\gamma)\left(r+\frac{1}{2\gamma}\left(\frac{\mu}{\sigma}\right)^{2} - \tilde{c}_{0}(m) - \frac{(1-\gamma)}{2\gamma^{2}}\left(\frac{\mu}{\sigma}\right)^{2}\right)t + \frac{(1-\gamma)\mu}{\gamma\sigma}B_{t}\right),
	\end{align}	
which is again a geometric Brownian motion that satisfies the dynamics 
	\[
	\frac{dY_t}{Y_t} = (1-\gamma)\left(r+\frac{1}{2\gamma}\left(\frac{\mu}{\sigma}\right)^{2} - \tilde{c}_{0}(m)\right)dt + \frac{(1-\gamma)\mu}{\gamma\sigma} dB_t,\quad Y_0=x^{1-\gamma}.
	\] 
Consequently, 
\begin{equation}\label{first part}
e^{-\Lambda^* t}	\E\left[e^{-\gamma(\psi-1)\frac{1-\zeta^{1-\gamma}}{1-\gamma}mt}(X^*_t)^{1-\gamma}\right]=x^{1-\gamma}e^{(C-\Lambda^*)t},
\end{equation}
where 
\[
C:= (1-\gamma)\left(r+\frac{1}{2\gamma}\left(\frac{\mu}{\sigma}\right)^{2} - \tilde{c}_{0}(m) \right)-\gamma(\psi-1)\frac{1-\zeta^{1-\gamma}}{1-\gamma}m. % = (1-\gamma)\psi\left(r+\frac{1}{2\gamma}\left(\frac{\mu}{\sigma}\right)^{2}-\delta\right) + \frac{\psi}{\theta}(\zeta^{1-\gamma}-1)m.
\]
Remarkably, by the definitions of $\tilde c_0(m)$ and $\Lambda^*$ in \eqref{c_0} and \eqref{Lambda^*}, a direct calculation shows that
$C-\Lambda^* = -\tilde c_0(m) <0$,  
where the inequality follows from $\tilde c_0(m)>0$. It follows from \eqref{first part} that 
\begin{equation}\label{de}
\lim_{t\to\infty} e^{-\Lambda^* t}	\E\left[e^{-\gamma(\psi-1)\frac{1-\zeta^{1-\gamma}}{1-\gamma}mt}(X^*_t)^{1-\gamma}\right]=0. 
\end{equation}
On the other hand, we can rewrite \eqref{X^*} as
\begin{align}\label{X^* with Z}
	(X^{*}_t)^{1-\gamma}= x^{1-\gamma}\text{exp}\left((1-\gamma)\left(r+\frac{1}{2\gamma}\left(\frac{\mu}{\sigma}\right)^{2} - \tilde{c}_{0}(m) \right)t \right)\cdot Z_t,
	\end{align}
where $Z$ is a geometric Brownian motion with the dynamics ${dZ_t} =Z_t \frac{(1-\gamma)\mu}{\gamma\sigma}dB_{t}$,  $Z_0=1$. As $Z$ is a martingale, we can apply the Burkh\"{o}lder-Davis-Gundy inequality to get
	\begin{equation}\label{BDG}
	\E\bigg[\sup\limits_{s\in[0,t]}(X_{s}^{*})^{1-\gamma}\bigg]\le K  x^{1-\gamma} e^{\left(|1-\gamma|\left|r+\frac{1}{2\gamma}\left(\frac{\mu}{\sigma}\right)^{2} - \tilde{c}_{0}(m) \right|\right) t} \frac{|1-\gamma|\mu}{\gamma\sigma} \E\left[\bigg(\int_0^t Z_s^2 ds\bigg)^{1/2}\right],
	\end{equation}
for some constant $K>0$. By Jensen's inequality and Fubini's theorem, 
\[
 \E\bigg[\bigg(\int_0^t Z_s^2 ds\bigg)^{1/2}\bigg] \le \bigg(\int_0^t \E[Z_s^2] ds\bigg)^{1/2} = \bigg(\int_0^t e^{\frac{(1-\gamma)^2\mu^2}{\gamma^2\sigma^2}s} ds\bigg)^{1/2} = \frac{\gamma\sigma}{|1-\gamma|\mu} \bigg(e^{\frac{(1-\gamma)^2\mu^2}{\gamma^2\sigma^2}t}-1\bigg)^{1/2}.
 \]
We then conclude from the above two inequalities that 
\begin{equation}\label{do}
\E\bigg[\sup\limits_{s\in[0,t]}(X_{s}^{*})^{1-\gamma}\bigg]<\infty,\quad \forall t\ge0.
\end{equation}
By \eqref{de} and \eqref{do}, $(X^{*})^{1-\gamma}$ satisfies \eqref{TransversalityCondition} (with $\Lambda^*$ in place of $\Lambda$), and so does the process $W^*_t:= w(X^*_t)=\delta^{\theta}\tilde{c}_{0}(m)^{-\frac{\theta}{\psi}}\frac{(X^*_t)^{1-\gamma}}{1-\gamma}$, i.e. $W^*\in \cE^{h^*}_{k^*}$. By applying It\^{o}'s formula to $W^*_t$ and noting
\begin{equation}\label{dodo}
\E\bigg[\sup_{s\in[0,t]}\pi^*_s(X_{s}^{*})^{1-\gamma}\bigg]<\infty\quad \hbox{for all}\ t\ge0,
\end{equation}
a consequence of \eqref{do} and $\pi^*_t\equiv \frac{\mu}{\gamma\sigma^2}$, we argue as in the proof of Theorem~\ref{Verification} that $W^*_t$ is a solution to \eqref{BSDE_V}. Moreover, %by direct calculation,
	\[
	W^*_t = \delta^\theta \tilde c_0(m)^{-\theta+(1-\gamma)}\frac{(X^*_t)^{1-\gamma}}{1-\gamma} = \delta^\theta \tilde c_0(m)^{-\theta}\frac{(c^*_t)^{1-\gamma}}{1-\gamma}	
	\]
By \eqref{c_0=k^*+...}, this shows that $W^*$ satisfies \eqref{comparison_condition} with $k=k^*$. Hence, $(c^*,h^*)$ is $k^*$-admissible, so that we can conclude $(c^*,\pi^*,h^*)\in\cP$. Theorem~\ref{Verification} is then applicable, asserting that $w(x,m)=v(x,m)$ and $(c^*,\pi^*,h^*)$ optimizes \eqref{problem'}. 	
	
%%%%%%%%%%%%%%%%%%%%%%%%%%%%%%%%%%		
	
\subsection{Proof of Theorem \ref{Thm:AH}}\label{subsec:proof of Thm:AH}
Define $w(x,m):=\delta^{\theta}\frac{x^{1-\gamma}}{1-\gamma}u^*(m)^{-\frac{\theta}{\psi}}$ for $(x,m)\in\R^2_+$. To apply Theorem~\ref{Verification}, we need to verify all its conditions. It can be checked, as in \eqref{candidates}-\eqref{ODE}, that $w\in C^{2,1}(\R_+\times\R_+)$ solves \eqref{HJB}. By the definitions of $\cP$ and $w$, $\cP\subseteq \cH_{k^*}$ and \eqref{veri condition} is satisfied for any $(c,\pi,h)\in\cP$. As $w_x>0$, $w_{xx}<0$, and $g$ satisfies Assumption~\ref{assump:AH}, $\bar c$, $\bar \pi$, and $\bar h$ in \eqref{bar's} are unique maximizers of the supermums in \eqref{sup1, 2} and \eqref{sup3}. It remain to show (i) for any $(c,\pi,h)\in\cP$,  $w(X^{c,\pi,h}_t, M^h_t)\in\cE^h_{k^*}$; (ii) $(c^*, \pi^*, h^*)$, defined using $\bar c$, $\bar\pi$, and $\bar h$ as in \eqref{optimal strategies}, belongs to $\cP$ and $W^*_t:= w(X^{c^*,\pi^*,h^*}_t,M^{h^*}_t)$ satisfies \eqref{comparison_condition}.  

{\bf (i)} Take any $\mathfrak p = (c,\pi,h)\in\cP$, and set $W_{t} := w(X_{t}^{\fp},M_{t}^{h})$ for $t\ge 0$. We will prove $W\in\cE^h_{k^*}$. %i.e. $W$ satisfies \eqref{TransversalityCondition} with $\Lambda$ therein taken as $\Lambda^*$ in \eqref{Lambda^*}.
\begin{itemize}[leftmargin=0.3in]
\item {\bf Case (i)-1:} $\gamma\in(\frac{1}{\psi},1)$. In view of \eqref{u^* bounds}, \eqref{c_0}, and \eqref{k^*}, we have $u^*(m)\ge \tilde c_0(m)\ge \tilde c_0(0)=k^*>0$. As $\theta>0$ when $\gamma\in(\frac{1}{\psi},1)$, this implies %Since $M^h_t = me^{\beta t}$ (as $g\equiv 0$), we have
\[
0<W_t=\delta^{\theta}\frac{(X^{\fp}_t)^{1-\gamma}}{1-\gamma}u^*(M^h_t)^{-\frac{\theta}{\psi}}\le \delta^{\theta}\frac{(X^{\fp}_t)^{1-\gamma}}{1-\gamma}(k^*)^{-\frac{\theta}{\psi}}\quad \forall t\ge 0,
\] 
Since $(X^{\fp})^{1-\gamma}$ satisfies \eqref{TransversalityCondition} (as $\fp\in\cP=\cP_1$), the above implies that $W$ also satisfies \eqref{TransversalityCondition}. 

\item{\bf Case (i)-2:} $\gamma>1$ and $\zeta<1$. As $\fp\in\cP=\cP_2$, there exists $\eta\in(1-\frac{1}{\gamma},1)$ such that \eqref{Permissible} holds. Consider %the constants
\begin{equation}\label{alpha's}
	\alpha := -\eta\frac{\gamma(\psi-1)}{1-\gamma}(\zeta^{1-\gamma}-1)>0,\qquad \alpha' := -(1-\eta)\frac{\gamma(\psi-1)}{1-\gamma}(\zeta^{1-\gamma}-1)>0,
\end{equation}	
%as well as the process
\begin{equation}\label{F}
F_t := \left(u_{\beta}(M_{t}^{h})\right)^{-\frac{\theta}{\psi}}\text{exp}\left(-\alpha'\int_{0}^{t}M_{s}^{h}ds\right)\quad \hbox{for}\ t\ge 0.
\end{equation}

First, we claim that the process $F$ is bounded from above; more specifically, 
\begin{equation}\label{F bounded}
\sup_{t\ge 0}F_t \le u_{\beta}\left(-\frac{\theta}{\alpha'\psi}\beta\right)^{-{\theta}/{\psi}}<\infty. %,\quad \forall t\ge 0. 
\end{equation}
%so that $\frac{\gamma(\psi-1)}{1-\gamma}(\zeta^{1-\gamma}-1) = -(\alpha' + \alpha)$. Let $F_t := \left(u_{\beta}(M_{t}^{h})\right)^{-\frac{\theta}{\psi}}\text{exp}\left(-\alpha'\int_{0}^{t}M_{s}^{h}ds\right)$, and 
Observe that
	\begin{align}
	\frac{dF_t}{dt} &= -\left(\alpha' M_{t}^{h} + \frac{\theta}{\psi}u_{\beta}(M_{t}^{h})^{-1}u_{\beta}'(M_{t}^{h})\frac{dM_{t}^{h}}{dt}\right)F_t\nonumber\\
	&= -\left(\alpha' M_{t}^{h} + \frac{\theta}{\psi\beta}(\beta-g(h_{t}))\big[u_{\beta}(M_{t}^{h}) - \tilde{c}_{0}(M_{t}^{h})\big]\right)F_t,\label{dF}
	\end{align}
	where the second equality follows as $u_\beta$ solves \eqref{ANHp_ODE} with $\ell=\beta$. For each $\omega\in\Omega$, consider %the collection of time points
	\[
	S(\omega) := \left\{t\ge 0: M_{t}^{h}(\omega)= \frac{-\theta}{\alpha'\psi\beta}(\beta-g(h_{t}))\big(u_{\beta}(M_{t}^{h}) - \tilde{c}_{0}(M_{t}^{h})\big)(\omega)\right\}. 
	\]
	%\[
	%\tau_{n}^{*}:= \inf\left\{t>\tau_{n-1}^{*}: M_{t}^{h}= \frac{-\theta}{\alpha'\psi\beta}(\beta-g(h_{t}))\big(u_{\beta}(M_{t}^{h}) - \tilde{c}_{0}(M_{t}^{h})\big)\right\},\quad n\in\N.
	%\]
	We deduce from \eqref{dF} that local maximizers of $t\mapsto F_t(\omega)$ must belong to $S(\omega)$, i.e.
	\begin{equation}\label{local max}
	\hbox{if $t\ge 0$ satisfies}\ F_t(\omega)= \max_{s\in[(t-\eps)^+,t+\eps]} F_s(\omega)\ \hbox{for some $\eps>0$, then $t\in S(\omega)$}.  
	\end{equation}
Also, by $g\ge 0$ and \eqref{ANHbound}, 
\begin{equation}\label{L<}
L_t (\omega):= \frac{-\theta}{\alpha'\psi\beta}(\beta-g(h_{t}))\big(u_{\beta}(M_{t}^{h}) - \tilde{c}_{0}(M_{t}^{h})\big)(\omega)\le -\frac{\theta}{\alpha'\psi}\beta,\quad \forall t\ge0.
\end{equation}
This particularly implies that 
\begin{equation}\label{M<}
M^h_t(\omega) = L_t(\omega)\le -\frac{\theta}{\alpha'\psi}\beta,\quad \hbox{for each}\ t\in S(\omega).
\end{equation}
Now, there are three distinct possibilities: 1) There exists $t^*\ge 0$ such that $M^h_t(\omega)< L_t(\omega)$ for all $t>t^*$. Then, $S(\omega)\subseteq [0,t^*]$ and \eqref{L<} implies $M^h_t(\omega)< -\frac{\theta}{\alpha'\psi}\beta$ for all $t> t^*$. It then follows from \eqref{local max} and \eqref{F} that
\begin{equation}\label{<t^*}
\sup_{t\le t^*} F_t(\omega)=  \sup_{t\in S(\omega)} F_t(\omega) \le  \sup_{t\in S(\omega)} u_{\beta}\big(M^h_t(\omega)\big)^{-\frac{\theta}{\psi}} \le u_{\beta}\left(-\frac{\theta}{\alpha'\psi}\beta\right)^{-{\theta}/{\psi}}, 
\end{equation}
where the last inequality follows from \eqref{M<}. Moreover, 
\[
\sup_{t> t^*} F_t(\omega)\le  \sup_{t> t^*} u_{\beta}\big(M^h_t(\omega)\big)^{-\frac{\theta}{\psi}}\le  u_{\beta}\left(-\frac{\theta}{\alpha'\psi}\beta\right)^{-{\theta}/{\psi}},
\]
i.e. \eqref{F bounded} holds. 2) There exists $t^*\ge 0$ such that $M^h_t(\omega)> L_t(\omega)$ for all $t>t^*$. By \eqref{dF}, $F_t(\omega)$ is strictly decreasing for $t>t^*$. Thus, $\sup_{t\ge 0} F_t(\omega) = \sup_{t\le t^*} F_t(\omega) = \sup_{t\in S(\omega)} F_t(\omega)$. By the estimate in \eqref{<t^*}, \eqref{F bounded} holds. 3) Neither 1) nor 2) above holds. This entails $\sup\{t\ge 0 : t\in S(\omega)\}=\infty$. Hence, $\sup_{t\ge 0} F_t(\omega) = \sup_{t\in S(\omega)} F_t(\omega)$, so that \eqref{F bounded} holds by the estimate in \eqref{<t^*}. 
Now, since $u^*\le u_{\beta}$ (by \eqref{u^* bounds}), $-\theta/\psi>0$, and $1-\gamma<0$, 
	\begin{align*}
	0&\ge e^{\frac{\gamma(\psi-1)}{1-\gamma}(\zeta^{1-\gamma}-1)\int_{0}^{t}M^h_{s}ds}W_{t}\ge \delta^{\theta}\left(u_{\beta}(M_{t}^{h})\right)^{-{\theta}/{\psi}}e^{\frac{\gamma(\psi-1)}{1-\gamma}(\zeta^{1-\gamma}-1)\int_{0}^{t}M^h_{s}ds}\frac{(X_{t}^{\fp})^{1-\gamma}}{1-\gamma}\\
	&=\delta^{\theta}F_t\ e^{-\alpha\int_{0}^{t}M^h_{s}ds}\frac{(X_{t}^{\fp})^{1-\gamma}}{1-\gamma}\ge \delta^{\theta}u_{\beta}\left(\frac{-\theta}{\alpha'\psi}\beta\right)^{-{\theta}/{\psi}}e^{-\alpha\int_{0}^{t}M^h_{s}ds}\frac{(X_{t}^{\fp})^{1-\gamma}}{1-\gamma},
	\end{align*}
	where the equality follows from \eqref{F} and \eqref{alpha's}, and the last inequality is due to \eqref{F bounded}. Recalling that $\fp\in\cP=\cP_{2}$, we conclude from \eqref{Permissible} and the above inequality that  
	\[
	\lim\limits_{t\rightarrow\infty}e^{-\Lambda^* t}\E\bigg[e^{\frac{\gamma(\psi-1)}{1-\gamma}(\zeta^{1-\gamma}-1)\int_{0}^{t}M_{s}^{h}ds}W_{t}\bigg] = 0.
	\]
	On the other hand, since $M_{t}^{h}\le m e^{\beta t}$,	
	\[
	\E\bigg[\sup\limits_{s\in[0,t]}|W_{t}|\bigg]\le \frac{\delta^{\theta}}{|1-\gamma|}u_\beta(me^{\beta t})^{-\theta/\psi} \E\bigg[\sup_{s\in[0,t]}(X_{s}^{\fp})^{1-\gamma}\bigg]<\infty,\quad \forall t\ge 0.
	\]
	where the finiteness is a direct consequence of $\fp\in \cP$. 
	
\item {\bf Case (i)-3:} $\gamma>1$ and $\zeta = 1$. In view of \eqref{u_q}, $u_\ell \equiv k^*>0$ for any $\ell>0$. It then follows from \eqref{u^* bounds} that $u^*\equiv k^*>0$. The required properties then follow directly from $\fp\in\cP=\cP_1$. 
\end{itemize}
	
{\bf (ii)} Now, we show that $(c^*,\pi^*,h^*)\in\cP$ and $W^*_t:= w(X^{c^*,\pi^*,h^*}_t,M^{h^*}_t)$ satisfies \eqref{comparison_condition}. Observe that a unique solution $M^*=M^{h^*}$ to \eqref{Mortality} exists. As $h^*$ by definition only depends on $u^*$, $g$, and the current mortality rate, $M^*$ is a deterministic process. Thanks to \eqref{g(h^*)<beta}, $t\mapsto M^*_t$ is strictly increasing. Also, a unique solution $X^*=X^{c^*,\pi^*,h^*}$ to \eqref{Wealth} exists, which admits the formula
	\begin{align}\label{X^*''}
	(X^{*}_t)^{1-\gamma}&=x^{1-\gamma}\exp\bigg(\int_{0}^{t}(1-\gamma)\left(r+\frac{1}{2\gamma}\left(\frac{\mu}{\sigma}\right)^{2} -u^*(M^{*}_s)-h_{s}^*- \frac{1-\gamma}{2\gamma^{2}}\left(\frac{\mu}{\sigma}\right)^{2}\right)ds%\nonumber\\&\hspace{4.25in}
	+ \frac{(1-\gamma)\mu}{\gamma\sigma}B_{t}\bigg).
	\end{align}
	%	\begin{equation}
	%(X^{*})^{1-\gamma}=x^{1-\gamma}\text{exp}\left(\int_{0}^{t}(1-\gamma)\left(r+\frac{1}{2\gamma}\left(\frac{\mu}{\sigma}\right)^{2} -(\tilde{u}(me^{\beta s})+h_{s}^*)- \frac{(1-\gamma)}{2\gamma^{2}}\left(\frac{\mu}{\sigma}\right)^{2}\right)ds + \frac{(1-\gamma)\mu}{\gamma\sigma}B_{t}\right).
	%\end{equation}
\begin{itemize}[leftmargin=0.3in]
\item {\bf Case (ii)-1:} $\gamma\in(\frac{1}{\psi},1)$.  As $M^*_t$ is strictly increasing, $u^*(M^*_t)\ge u^*(m)\ge \tilde c_0(m)$, where the second inequality follows from \eqref{u^* bounds} and \eqref{ANHbound}. With this and $h^*_{t}\ge0$, we deduce from \eqref{X^*''} that \eqref{X^*} holds with ``$=$'' therein replaced by ``$\le$''. 
	As $k^*>0$ entails $\tilde c_0(m)>0$ (see \eqref{c_0=k^*+...}), the same arguments in Proposition \ref{prop:NoAging} can be applied to  to show that $(X^*)^{1-\gamma}$ satisfies \eqref{TransversalityCondition}. With this, we can argue as in Case (i)-1 to show that $W^*_t:= w(X^*_t,M^*_t)$ belongs to $\cE^{h^*}_{k^*}.$

\item {\bf Case (ii)-2:} $\gamma>1$ and $\zeta< 1$. As $u^*$ solves \eqref{ODE} and $h^*$ maximizes the supremum in \eqref{ODE},
	\begin{align*}
	u^*(M_{t}^{*})-\tilde{c}_{0}(M_{t}^{*})-(\psi-1)h^*_t=\frac{M_{t}^{*}(u^*)'(M_{t}^{*})}{u^*(M_{t}^{*})}(\beta-g(h^*_t)) >0\quad \forall t>0, 
	\end{align*}
	where the inequality follows from \eqref{g(h^*)<beta}. This gives $h^*_t< \frac{1}{\psi-1}(u^*(M_{t}^{*})-\tilde{c}_{0}(M_{t}^{*}))$, so that
	\begin{equation}\label{B4}
	u^*(M_{t}^{*})+h^*_t <\frac{\psi}{\psi-1}u^*(M_{t}^{*}) - \frac{1}{\psi-1}\tilde{c}_{0}(M_{t}^{*})\le \frac{\psi}{\psi-1}u_\beta(M_{t}^{*}) - \frac{1}{\psi-1}\tilde{c}_{0}(M_{t}^{*}),
	\end{equation}
where the last inequality follows from $u^*(m)\le {u}_{\beta}(m)$ (see \eqref{u^* bounds}).  %	Using this with bound \eqref{B4} yields
%	\begin{align*}
%	&(X^{*})^{1-\gamma}\le\\ &x^{1-\gamma}\text{exp}\left(\int_{0}^{t}(1-\gamma)\left(r+\frac{1}{2\gamma}\left(\frac{\mu}{\sigma}\right)^{2} - \frac{\psi}{\psi-1}\tilde{u}(M_{s}^*) +\frac{1}{\psi-1}\tilde{c}_{0}(M_{s}^*)- \frac{(1-\gamma)}{2\gamma^{2}}\left(\frac{\mu}{\sigma}\right)^{2}\right)ds + \frac{(1-\gamma)\mu}{\gamma\sigma}B_{t}\right)\\
%	&x^{1-\gamma}\text{exp}\left(\int_{0}^{t}(1-\gamma)\left(r+\frac{1}{2\gamma}\left(\frac{\mu}{\sigma}\right)^{2}+k - \frac{\psi}{\psi-1}\left(\tilde{u}(M_{s}^*)-\frac{1}{\theta}M_{s}^*\right)- \frac{(1-\gamma)}{2\gamma^{2}}\left(\frac{\mu}{\sigma}\right)^{2}\right)ds + \frac{(1-\gamma)\mu}{\gamma\sigma}B_{t}\right)
%	\end{align*}
	For any $\eta\in(1-\frac{1}{\gamma},1)$, consider $\alpha,\alpha'>0$ defined as in \eqref{alpha's}. Observe that ${u}_\beta(m)$ can be written as
	\[
	{u}_\beta(m) = \beta\frac{e^{-m\frac{\psi}{\theta\beta}(1-\zeta^{1-\gamma})}\left(m\frac{\psi}{\theta\beta}(1-\zeta^{1-\gamma})\right)^{-{k^*}/{\beta}}}{\overline{\Gamma}\left(-\frac{k^*}{\beta},m\frac{\psi}{\theta\beta}(1-\zeta^{1-\gamma})\right)}
	\]
where $\overline{\Gamma}$ is the upper incomplete gamma function $\overline{\Gamma}(s,z):=\int_{z}^{\infty}t^{s-1}e^{-t}dt$. Similarly to the argument in \cite[(A.6)-(A.7)]{Huang19}, by using the fact $\lim_{z\rightarrow\infty}\frac{\overline{\Gamma}(s,z)}{e^{-z}z^{s-1}} =1$, 
	\begin{equation}\label{lim>1}
	\lim\limits_{m\rightarrow\infty}\frac{\psi-1}{\psi}\frac{\left(\alpha+(\zeta^{1-\gamma}-1)\right) m}{(\gamma-1){u}_\beta(m)} = \frac{\psi-1}{\psi}\frac{\alpha+(\zeta^{1-\gamma}-1)}{(\psi-1)(\zeta^{1-\gamma}-1)} = \frac{\alpha+(\zeta^{1-\gamma}-1)}{\psi(\zeta^{1-\gamma}-1)}>1,
	\end{equation}
	where the inequality follows from the definition of $\alpha$ and $\eta>1-\frac{1}{\gamma}$. This, together with $M^*$ being a strictly increasing deterministic process, implies the existence of $s^{*}>0$ such that 
	\begin{equation}\label{>for s>}
	(\alpha +(\zeta^{1-\gamma}-1))M_{s}^* > \frac{\psi(\gamma-1)}{\psi-1}{u}_\beta(M_{s}^*)\quad \hbox{for}\ s> s^{*}.
	\end{equation}
		Consider the constant
	$
	0\le K:= \max_{t\in[0,s^*]}\big\{ \frac{\psi}{\psi-1}{u}_\beta(M^*_t) - \frac{\alpha +(\zeta^{1-\gamma}-1)}{\gamma-1}M^*_t\big\}<\infty.
	$
In view of \eqref{X^*''}, \eqref{B4}, and $\tilde c_0(m)=k^*+(1-\psi)\frac{\zeta^{1-\gamma}-1}{1-\gamma}m$ (see 
\eqref{c_0=k^*+...}), 
	\begin{align*}
	&e^{-\alpha\int_{0}^{t}M^*_{s}ds}(X^{*}_t)^{1-\gamma}\\
	&\le x^{1-\gamma}\text{exp}\left(\int_{0}^{t}(1-\gamma)\left(r+\frac{1}{2\gamma}\left(\frac{\mu}{\sigma}\right)^{2}+\frac{k^*}{\psi-1} - \frac{\psi}{\psi-1}{u}_\beta(M_{s}^*)-\frac{\alpha +(\zeta^{1-\gamma}-1)}{1-\gamma} M_{s}^*\right)ds\right)\cdot Z_t\\
	&\le x^{1-\gamma} e^{(1-\gamma)\left(r+\frac{1}{2\gamma}\left(\frac{\mu}{\sigma}\right)^{2} + \frac{k^*}{\psi-1} - K\right) s^*} e^{(1-\gamma)\left(r+\frac{1}{2\gamma}\left(\frac{\mu}{\sigma}\right)^{2} + \frac{k^*}{\psi-1} \right) (t-s^*)} Z_t,
	\end{align*}
	where $Z$ is the driftless geometric Brownian motion defined below \eqref{X^* with Z}, and the second inequality follows from \eqref{>for s>}. It follows that
	\begin{align*}
	&e^{-\Lambda^* t}\E\left[e^{-\alpha\int_{0}^{t}M^*_{s}ds}(X_{t}^{*})^{1-\gamma}\right]\\
	 &\le  x^{1-\gamma} e^{\left((1-\gamma)\left(r+\frac{1}{2\gamma}\left(\frac{\mu}{\sigma}\right)^{2} +\frac{k^*}{\psi-1} - K\right)-\Lambda^*\right) s^*}  e^{\left((1-\gamma)\left(r+\frac{1}{2\gamma}\left(\frac{\mu}{\sigma}\right)^{2}+ \frac{k^*}{\psi-1} \right)-\Lambda^*\right)(t-s^*)}\\
	&= x^{1-\gamma} e^{\left((1-\gamma)\left(r+\frac{1}{2\gamma}\left(\frac{\mu}{\sigma}\right)^{2} + \frac{k^*}{\psi-1} - K\right)-\Lambda^*\right) s^*}  e^{-(\gamma+\frac{\gamma-1}{\psi-1}) k^*(t-s^*)}\to\ 0\quad \hbox{as}\ t\to\infty,
	\end{align*}
	where the equality follows from a direct calculation using the definition of $\Lambda^*$ in \eqref{Lambda^*}, and the convergence is due to $k^*>0$. Namely, $X^*$ satisfies \eqref{Permissible}. On the other hand, by \eqref{B4} and $M^*_t\le me^{\beta t}$, we obtain from \eqref{X^*''} that %similarly to \eqref{X^* with Z}, we can rewrite \eqref{X^*'} as
\begin{align*}%\label{X^* with Z}
	(X^{*}_t)^{1-\gamma}\le x^{1-\gamma}\text{exp}\left(\int_0^t (1-\gamma)\left(r+\frac{1}{2\gamma}\left(\frac{\mu}{\sigma}\right)^{2} - \frac{\psi}{\psi-1}{u}_\beta(me^{\beta s}) \right) ds\right)\cdot Z_t,
	\end{align*}
where $Z$ is again the driftless geometric Brownian motion defined below \eqref{X^* with Z}. By the Burkh\"{o}lder-Davis-Gundy inequality, we obtain the estimate in \eqref{BDG} with $-\tilde c_0(m)$ therein replaced by $\frac{\psi}{\psi-1}{u}_\beta(me^{\beta t})$. This then implies $\E\big[\sup_{s\in[0,t]}(X_{s}^{*})^{1-\gamma}\big]<\infty$, by the inequality preceding \eqref{do}. 
Finally, under $\E\big[\sup_{s\in[0,t]}(X_{s}^{*})^{1-\gamma}\big]<\infty$ and \eqref{Permissible}, the same argument as in Case (i)-2 shows that $W^*_t:= w(X^*_t,M^*_t)$ belongs to $\cE^{h^*}_{k^*}.$ 
		
\item {\bf Case (ii)-3:} $\gamma>1$ and $\zeta = 1$. By \eqref{u_q}, ${u}_\beta(m) \equiv k^*>0$. As $M_{t}^*$ is strictly increasing, $\tilde c_0(M^*_t)\ge \tilde c_0(0)=k^*$.  The estimate \eqref{B4} then becomes 
	%\begin{equation}\label{B6}
	$u^*(M_{t}^{*})+h^*\le \frac{\psi}{\psi-1}k^* - \frac{1}{\psi-1}k^* = k^*$,
	%\end{equation}
	so that we can deduce from \eqref{X^*''} that
	\begin{align*}
	(X^{*}_t)^{1-\gamma}\le x^{1-\gamma}\text{exp}\left(\int_{0}^{t}(1-\gamma)\left(r+\frac{1}{2\gamma}\left(\frac{\mu}{\sigma}\right)^{2} -k^*- \frac{(1-\gamma)}{2\gamma^{2}}\left(\frac{\mu}{\sigma}\right)^{2}\right)ds + \frac{(1-\gamma)\mu}{\gamma\sigma}B_{t}\right).
	\end{align*}
	The arguments in Proposition \ref{prop:NoAging} can then be applied to show that $(X^*)^{1-\gamma}$ satisfies \eqref{TransversalityCondition}. Then, we may argue as in Case (i)-3 to show that $W^*_t:= w(X^*_t,M^*_t)$ belongs to $\cE^{h^*}_{k^*}.$ 
\end{itemize}
%Proposition \ref{prop:Aging}
	Finally, by applying It\^{o}'s formula to $W^*_t$ and using \eqref{dodo}, a consequence of \eqref{do} and $\pi^*_t\equiv \frac{\mu}{\gamma\sigma^2}$, we argue as in the proof of Theorem~\ref{Verification} that $W^*_t$ is a solution to \eqref{BSDE_V}. Also,
	\begin{equation}\label{W*<}
	W^*_t = \delta^\theta u^*(M^*_t)^{-\theta+(1-\gamma)} \frac{(X^*_t)^{1-\gamma}}{1-\gamma} = \delta^\theta u^*(M^*_t)^{-\theta}\frac{(c^*_t)^{1-\gamma}}{1-\gamma}\le\delta^\theta \tilde c_0(M^*_t)^{-\theta}\frac{(c^*_t)^{1-\gamma}}{1-\gamma}, 
	\end{equation}
	where the inequality follows from $u^*\ge \tilde c_0$ (by \eqref{u^* bounds} and \eqref{ANHbound}) and the fact that $\theta>0$ if $\gamma\in(\frac{1}{\psi},1)$ and $\theta<0$ if $\gamma>1$. By \eqref{c_0=k^*+...}, this shows that $W^*$ satisfies \eqref{comparison_condition} with $k=k^*$. Hence, $(c^*,h^*)$ is $k^*$-admissible, and we can now conclude $(c^*,\pi^*,h^*)\in\cP$. By Theorem \ref{Verification}, $v(x,m)=w(x,m)$ and $(c^*,\pi^*,h^*)$ optimizes \eqref{problem'}.

\subsection{Proof of Proposition \ref{prop:Aging}}\label{subsec:proof of prop:Aging}
Define $w(x,m):=\delta^{\theta}\frac{x^{1-\gamma}}{1-\gamma}u_\beta(m)^{-\frac{\theta}{\psi}}$ for $(x,m)\in\R^2_+$. To apply Theorem~\ref{Verification}, we need to verify all its conditions. It can be checked directly that $w\in C^{2,1}(\R_+\times\R_+)$  solves \eqref{HJB}, as $u_\beta$ is a solution to \eqref{ODE'} (Lemma~\ref{ANH_ODE}). By the definitions of $\cP$ and $w$, $\cP\subseteq \cH_{k^*}$ and \eqref{veri condition} is satisfied for any $(c,\pi,h)\in\cP$. 
Following part (i) of the proof of Theorem~\ref{Thm:AH}, we get $w(X^{c,\pi,h}_t, M^h_t)\in\cE^h_{k^*}$ for any $(c,\pi,h)\in\cP$; the proof is much simpler here, as $M^h_t=me^{\beta t}$ in the current setting. As $w_x>0$, $w_{xx}<0$, $\bar c(x,m):=x u_\beta(m)$ and $\bar \pi(x,m) :=\frac{\mu}{\gamma\sigma^2}$ are unique maximizers of the supremums in \eqref{sup1, 2}, respectively. The supremum in \eqref{sup3} is zero, as $g\equiv 0$ and $w_x>0$. Hence, $\bar h(x,m) := 0$ trivially maximizes  \eqref{sup3}. It remains to show that $(c^*, \pi^*, h^*)$, defined using $\bar c$, $\bar\pi$, and $\bar h$ as in \eqref{optimal strategies}, belongs to $\cP$ and $W^*_t:= w(X^{c^*,\pi^*,h^*}_t,M^{h^*}_t)$ satisfies \eqref{comparison_condition}. 
	
Observe that $M^{h^*}_t = me^{\beta t}$ as $h^*\equiv 0$, and a unique solution $X^*=X^{c^*,\pi^*,h^*}$ to \eqref{Wealth} exists, which satisfies the dynamics \eqref{X^* GBM} with $\tilde c_0(m)$ replaced by ${u}_\beta(me^{\beta t})$. 
	%First calculate the dynamics of $X_{t}^*$ via $$\frac{dX_{t}^{*}}{X_{t}^{*}} = \left(r+\frac{1}{\gamma}\left(\frac{\mu}{\sigma}\right)^{2} - \tilde{u}(me^{\beta t})\right)dt + \frac{\mu}{\gamma\sigma}dB_{t}$$
	This implies
	\begin{equation}\label{X^*'}
		(X^{*}_t)^{1-\gamma}= x^{1-\gamma}\text{exp}\left(\int_{0}^{t}(1-\gamma)\left(r+\frac{1}{2\gamma}\left(\frac{\mu}{\sigma}\right)^{2} - {u}_\beta(me^{\beta s}) - \frac{(1-\gamma)}{2\gamma^{2}}\left(\frac{\mu}{\sigma}\right)^{2}\right)ds + \frac{(1-\gamma)\mu}{\gamma\sigma}B_{t}\right).
	\end{equation}
\begin{itemize}[leftmargin=0.3in]
\item	{\bf Case 1:} $\gamma\in(\frac{1}{\psi},1)$. As $1-\gamma>0$ and $u_\beta(m) \ge \tilde c_0(m)$ (see \eqref{ANHbound}),  we deduce from \eqref{X^*'} that \eqref{X^*} holds with ``$=$'' therein replaced by ``$\le$''. As $k^*>0$ entails $\tilde c_0(m)>0$, the same arguments in Proposition \ref{prop:NoAging} can be applied to show that $(X^*)^{1-\gamma}$ satisfies \eqref{TransversalityCondition}. With this, we can argue as in Case (i)-1 of the proof of Theorem~\ref{Thm:AH} to obtain $W^*_t:= w(X^*_t,M^*_t)\in \cE^{h^*}_{k^*}.$ 
	
\item	{\bf Case 2:} $\gamma>1$ and $\zeta\neq 1$.
 For any $\eta\in(1-\frac{1}{\gamma},1)$, consider the constant $\alpha>0$ defined in \eqref{alpha's}. %$\alpha := -\frac{\gamma(\psi-1)}{1-\gamma}(\zeta^{1-\gamma}-1)>0$. 
Similarly to \eqref{lim>1}, using the fact that $\lim_{z\rightarrow\infty}\frac{\overline{\Gamma}(s,z)}{e^{-z}z^{s-1}} =1$ yields 
	\begin{equation}\label{lim>1'}
	\lim\limits_{m\rightarrow\infty}\frac{\alpha m}{(\gamma-1)\tilde{u}(m)} = \frac{\alpha}{(\psi-1)(\zeta^{1-\gamma}-1)} >1.
	\end{equation}
	where the inequality follows from the definition of $\alpha$ and $\eta>1-\frac{1}{\gamma}$.  %$\eta>1-\frac{1}{\gamma}$ ensures $\alpha > (\psi-1)(\zeta^{1-\gamma}-1)$. As $M_{t}$ is  strictly increasing, 
	This implies that there exists some $s^{*}>0$ such that 
	\begin{equation}\label{M dominates}
	\alpha me^{\beta s} \ge (\gamma-1)\tilde{u}(me^{\beta s})\quad \hbox{for all $s\ge s^{*}$}. 
	\end{equation}
	Consider
	$
	0\le K:= \max_{t\in[0,s^*]}\left\{\tilde{u}(me^{\beta t}) - \frac{\alpha m e^{\beta t}}{\gamma-1}\right\}<\infty.
	$
	Now, by $M_t=me^{\beta t}$ and \eqref{X^*'} ,
		\begin{align*}
	e^{-\alpha\int_{0}^{t}M_{s}ds}(X^{*})^{1-\gamma}&=x^{1-\gamma}\exp\left(\int_{0}^{t}(1-\gamma)\left(r+\frac{1}{2\gamma}\left(\frac{\mu}{\sigma}\right)^{2} - {u}_\beta(me^{\beta s}) -\frac{\alpha me^{\beta s}}{(1-\gamma)} \right)ds \right)\cdot Z_t\\
&\le x^{1-\gamma} e^{(1-\gamma)\left(r+\frac{1}{2\gamma}\left(\frac{\mu}{\sigma}\right)^{2} - K\right)s^*}  e^{(1-\gamma)\left(r+\frac{1}{2\gamma}\left(\frac{\mu}{\sigma}\right)^{2}\right)(t-s^*)}Z_t,
	\end{align*}
	where $Z_t$ is the driftless geometric Brownian motion defined below \eqref{X^* with Z}, and the inequality follows from \eqref{M dominates}. It follows that 
	\begin{align*}
	e^{-\Lambda^* t}\E\left[e^{-\alpha\int_{0}^{t}M_{s}ds}(X_{t}^{*})^{1-\gamma}\right] &\le  x^{1-\gamma} e^{\left((1-\gamma)\left(r+\frac{1}{2\gamma}\left(\frac{\mu}{\sigma}\right)^{2} - K\right)-\Lambda^*\right) s^*}  e^{\left((1-\gamma)\left(r+\frac{1}{2\gamma}\left(\frac{\mu}{\sigma}\right)^{2}\right)-\Lambda^*\right)(t-s^*)}\\
	&= x^{1-\gamma} e^{\left((1-\gamma)\left(r+\frac{1}{2\gamma}\left(\frac{\mu}{\sigma}\right)^{2} - K\right)-\Lambda^*\right) s^*}  e^{-\gamma k^*(t-s^*)}\to\ 0,\quad \hbox{as}\ t\to\infty, 
	\end{align*}
	where the second line follows from a direct calculation using the definition of $\Lambda^*$ in \eqref{Lambda^*}, and the convergence is due to $k^*>0$.  
	On the other hand, similarly to \eqref{X^* with Z}, we rewrite \eqref{X^*'} as
\begin{align*}%\label{X^* with Z}
	(X^{*}_t)^{1-\gamma}= x^{1-\gamma}\text{exp}\left(\int_0^t (1-\gamma)\left(r+\frac{1}{2\gamma}\left(\frac{\mu}{\sigma}\right)^{2} - {u}_\beta(me^{\beta s}) \right) ds\right)\cdot Z_t,
	\end{align*}
where $Z$ is again the driftless geometric Brownian motion defined below \eqref{X^* with Z}. By Burkh\"{o}lder-Davis-Gundy's inequality, we obtain the estimate in \eqref{BDG} with $-\tilde c_0(m)$ therein replaced by ${u}_\beta(me^{\beta t})$. This implies $\E\big[\sup_{s\in[0,t]}(X_{s}^{*})^{1-\gamma}\big]<\infty$, by the inequality preceding \eqref{do}. 
Under $\E\big[\sup_{s\in[0,t]}(X_{s}^{*})^{1-\gamma}\big]<\infty$ and \eqref{Permissible}, the same argument as in Case (i)-2 of the proof of Theorem~\ref{Thm:AH} shows that $W^*_t:= w(X^*_t,M^*_t)$ belongs to $\cE^{h^*}_{k^*}.$ 

\item {\bf Case 3:} $\gamma>1$ and $\zeta=1$. By \eqref{u_q}, ${u}_\beta(m) \equiv k^*>0$. Then, in view of \eqref{X^*'}, we can apply the same arguments as in Proposition \ref{prop:NoAging} to show that $(X^*)^{1-\gamma}$ satisfies \eqref{TransversalityCondition}. With this, we may argue as in Case (i)-3 in the proof of Theorem~\ref{Thm:AH} to obtain $W^*_t:= w(X^*_t,M^*_t)\in\cE^{h^*}_{k^*}.$ 
\end{itemize}
Finally, by applying It\^{o}'s formula to $W^*_t$ and using \eqref{dodo}, a consequence of \eqref{do} and $\pi^*_t\equiv \frac{\mu}{\gamma\sigma^2}$, we argue as in the proof of Theorem~\ref{Verification} that $W^*_t$ is a solution to \eqref{BSDE_V}. Also, the same calculation as in \eqref{W*<}, with $u^*$ therein replaced by $u_\beta$, can be carried out, 
	%\[
	%W^*_t = \delta^\theta u_\beta(M^*_t)^{-\theta+(1-\gamma)} \frac{(X^*_t)^{1-\gamma}}{1-\gamma} = \delta^\theta u_\beta(M^*_t)^{-\theta}\frac{(c^*_t)^{1-\gamma}}{1-\gamma}\le\delta^\theta \tilde c_0(M^*_t)^{-\theta}\frac{(c^*_t)^{1-\gamma}}{1-\gamma}, 
	%\]
	%where the inequality follows from 
	thanks to $u_\beta \ge \tilde c_0$ by \eqref{ANHbound}.  %and the fact that $\theta>0$ if $\gamma\in(\frac{1}{\psi},1)$ and $\theta<0$ if $\gamma>1$. 
	%By \eqref{c_0=k^*+...}, 
	This shows that $W^*$ satisfies \eqref{comparison_condition} with $k=k^*$. Hence, $(c^*,h^*)$ is $k^*$-admissible, and we can conclude $(c^*,\pi^*,h^*)\in\cP$. By Theorem \ref{Verification}, $v(x,m)=w(x,m)$ and $(c^*,\pi^*,h^*)$ optimizes \eqref{problem'}. 			
	
%%%%%%%%%%%%%%%%%%%%%

\small		
\bibliographystyle{siam}
\bibliography{refsMathSci}
\end{document}